\documentclass[11pt,tightenlines,notitlepage,nofootinbib]{revtex4-1}

\usepackage[latin1]{inputenc}
\usepackage[T1]{fontenc}
\usepackage{amsmath,amsfonts,amssymb,stmaryrd,amsthm,url,geometry,mathrsfs,color,esint,bm,tikz}
\usepackage{mathtools}
\usepackage{braket,verbatim}

\newcommand{\e}{\varepsilon}
\newcommand{\iy}{\infty}
\newcommand{\st}{\ : \ }

\renewcommand{\leq}{\leqslant}
\renewcommand{\geq}{\geqslant}

\newcommand{\gM}{\mathsf{M}}

\newcommand{\gL}{\mathsf{L}}
\newcommand{\gPos}{\mathsf{Pos}}

\renewcommand{\P}{\mathbf{P}}
\newcommand{\R}{\mathbf{R}}
\newcommand{\C}{C} 
\newcommand{\CC}{\mathscr{C}}

\newcommand{\Id}{\mathrm{Id}}
\newcommand{\Herm}{\mathbb{H}}
\newcommand{\PSD}{\mathrm{PSD}}

\DeclareMathOperator{\tr}{Tr}

\DeclareMathOperator{\conv}{\mathrm{conv}}

\DeclareMathOperator{\area}{\mathrm{area}}

\DeclareMathOperator{\mathspan}{\mathrm{span}}

\DeclareMathOperator{\aff}{\mathrm{aff}}
\DeclareMathOperator{\inter}{int}

\newcommand{\tmin}{\varodot}
\newcommand{\tmax}{\varoast}

\newcommand{\texteq}[1]{\stackrel{\mathclap{\scriptsize \mbox{#1}}}{=}}
\newcommand{\textleq}[1]{\stackrel{\mathclap{\scriptsize \mbox{#1}}}{\leq}}
\newcommand{\textgeq}[1]{\stackrel{\mathclap{\scriptsize \mbox{#1}}}{\geq}}
\newcommand{\textsubseteq}[1]{\stackrel{\mathclap{\scriptsize \mbox{#1}}}{\subseteq}}

\newcommand\floor[1]{\left\lfloor#1\right\rfloor}

\newcommand{\sumno}{\sum\nolimits}
\newcommand{\erob}{E_{\mathrm{rob}}}

\DeclareMathOperator{\cone}{\mathrm{cone}}

\DeclareMathOperator{\relint}{relint}

\theoremstyle{plain}
\newtheorem{theorem}{Theorem}
\newtheorem{proposition}[theorem]{Proposition}
\newtheorem{lemma}[theorem]{Lemma}

\newtheorem{corollary}[theorem]{Corollary}

\newtheorem{result}[theorem]{Result}
\theoremstyle{definition}
\newtheorem{definition}{Definition}
\newtheorem{fact}[definition]{Fact}

\newtheorem*{note}{Note}
\newtheorem*{example}{Example}

\newcommand{\tcb}[1]{{\color{blue} #1}}

\usepackage{enumerate}
\usepackage[shortlabels]{enumitem}

\newtheorem{manualthminner}{Theorem}
\newenvironment{manualthm}[1]{%
  \renewcommand\themanualthminner{#1}%
  \manualthminner
}{\endmanualthminner}

\newtheorem{manualpropinner}{Proposition}
\newenvironment{manualprop}[1]{%
  \renewcommand\themanualpropinner{#1}%
  \manualpropinner
}{\endmanualpropinner}

\newtheorem{manuallemmainner}{Lemma}
\newenvironment{manuallemma}[1]{%
  \renewcommand\themanuallemmainner{#1}%
  \manuallemmainner
}{\endmanuallemmainner}

\begin{document}

\author{Guillaume Aubrun}
\email{aubrun@math.univ-lyon1.fr}
\affiliation{Institut Camille Jordan, Universit\'e Claude Bernard Lyon 1, 43 boulevard du 11 novembre 1918, 69622 Villeurbanne CEDEX, France}

\author{Ludovico Lami}
\email{ludovico.lami@gmail.com}
\affiliation{School of Mathematical Sciences and Centre for the Mathematics and Theoretical Physics of Quantum Non-Equilibrium Systems, University of Nottingham, University Park, Nottingham NG7 2RD, United Kingdom}
\affiliation{Institute of Theoretical Physics and IQST, Universit\"{a}t Ulm, Albert-Einstein-Allee 11D-89069 Ulm, Germany}

\author{Carlos Palazuelos}
\email{cpalazue@ucm.es}
\affiliation{Departamento de An\'alisis Matem\'atico y Matem\'atica Aplicada, Universidad Complutense de Madrid, Plaza de Ciencias
s/n 28040 Madrid, Spain,}
\affiliation{Instituto de Ciencias Matem\'aticas, C/ Nicol\'as Cabrera, 13-15, 28049
Madrid, Spain}

\title{Universal entangleability of non-classical theories}

\begin{abstract}
Inspired by its fundamental importance in quantum mechanics, we define and study the notion of entanglement for abstract physical theories, investigating its profound connection with the concept of superposition. We adopt the formalism of general probabilistic theories (GPTs), encompassing all physical models whose predictive power obeys minimal requirements. Examples include classical theories, which do not exhibit superposition and whose state space has the shape of a simplex, quantum mechanics, as well as more exotic models such as Popescu--Rohrlich boxes. We call two GPTs entangleable if their composite admits either entangled states or entangled measurements, and conjecture that any two non-classical theories are in fact entangleable. We present substantial evidence towards this conjecture by proving it (1)~for the simplest case of $3$-dimensional theories; (2)~when the local state spaces are discrete, which covers foundationally relevant cases; (3)~when one of the local theories is quantum mechanics. Furthermore, (4)~we envision the existence of a quantitative relation between local non-classicality and global entangleability, explicitly describing it in the geometrically natural case where the local state spaces are centrally symmetric.
\end{abstract}

\maketitle


The discovery that physical systems can be entangled can be deemed one of the main scientific achievements of the past century. While entanglement emerges naturally as a mathematical by-product of the formalism of quantum mechanics, its status has been the subject of an intense debate, with Einstein famously seeing it as the cause of the `spooky action at a distance' that inexorably affects the theory~\cite{EPR}. It was not until the work of Bell~\cite{Bell} that entanglement was promoted from a mere manifestation of the alleged incompleteness of quantum mechanics to a fully-fledged physical phenomenon, whose ultimate consequences for the non-locality of physics~\cite{Brunner-review} are testable and turn out to be rooted in experimental evidence~\cite{Aspect}. We nowadays conceive and study it as a fundamental feature that separates quantum and classical theory~\cite{Horodecki-review}.

Yet comparing its theoretical and empirical status, we note a striking difference: while entanglement is regarded as a purely quantum phenomenon on the theoretical level, the fact that we see non-locality in experiments implies that it must characterise every successful future theory of Nature, thus suggesting that its conceptual importance goes well beyond present-day quantum models. We however seem to lack a precise understanding of what it may even mean outside the well-studied quantum formalism. This is to be contrasted for instance with the satisfying model-independent definition that we have for non-locality~\cite{Brunner-review}.

In this paper we present a unified theory of \emph{universal entanglement} that treats it in a fully model-independent fashion. We look at a general definition of entangled state in a general bipartite physical system, arguing that any experimentally detectable non-local effect must come from some form of entanglement. We then proceed to study the connection between two seemingly different yet somehow intimately connected features of physical theories: their deviation from classicality\footnote{Here, a system is said to be \emph{classical} if there is a finite number of special pure states, and every other state can be thought of as resulting from a uniquely defined statistical ensemble of those pure states.} at the single-system level on the one hand, and their \emph{entangleability} at the level of bipartite systems on the other. Physicists have long sensed that these concepts may be related, as the example of quantum mechanics shows -- coherent superpositions of orthogonal pure states are the distinct signature of quantumness of a single system, and lead directly to quantum entanglement when constructed from product states. However, only in recent times this connection has been the subject of systematic investigation. Equipped with our rigorous theory of universal entanglement, we show that classicality and entangleability are directly related concepts at a foundational level. More precisely, we prove that under certain natural assumptions every system composed of a pair of non-classical models admits either entangled states or entangled measurements. Our results cover but are not limited to the cases of both state spaces being discrete, i.e.\ admitting only a finite number of pure states, or satisfying some regularity conditions. We conjecture that analogous statements would hold without these assumptions for the most general pair of non-classical theories. 
Our analysis shows that foundational questions such as those concerning the a priori role of entanglement in physical systems are -- rather surprisingly -- deeply rooted in convex geometry and functional analysis.



\section{Results}

\subsection{General probabilistic theories}

The question of what features a set of axioms should possess in order to be considered a fully-fledged physical theory has become particularly controversial with the rise of quantum theory. Without going too much into the philosophical aspects of the problem, a minimalist answer is as follows: \emph{a physical theory is a set of rules that allow to deduce a probabilistic prediction of the outcome of an experiment given the detailed description of its preparation.} Remarkably, it is possible to translate this general idea into rigorous axioms, from which a unified theoretical framework can be derived~\cite{Ludwig-1, LUDWIG, lamiatesi}.
The resulting formalism of \emph{general probabilistic theories} (GPTs) encompasses classical probability theory and quantum mechanics as special cases, but includes also a wealth of other models that may or may not be relevant for future physics. We now set out to describe briefly the setup, referring the reader to the many presentations available in the literature for further details~\cite{telep-in-GPT, Pfister-no-disturbance, lamiatesi}.

The basic ingredient is a \emph{state space}, i.e.\ a convex and compact subset $\Omega$ of some finite-dimensional real vector space. Throughout this paper we will always make the exquisitely technical assumption of finite dimension; extending the theory to infinite dimension does not require any conceptual modification yet makes it significantly more involved~\cite[Chapter~1]{lamiatesi}. States, that is, preparation procedures for a given physical system, are represented by points $\omega\in\Omega$. The convexity of $\Omega$ serves to model the existence of stochastic preparation procedures: flipping a coin with outcome probabilities $p$ (head) and $1-p$ (tail), preparing the system according to $\omega$ (head) or $\tau$ (tail), and subsequently forgetting the outcome of the coin should result in the system being in a state $p\omega + (1-p) \tau\in\Omega$.

For reasons that will become clear soon, it is useful to enlarge the vector space where $\Omega$ lives, increasing its dimension by one. The resulting vector space $V$ can be thought of as comprising all real multiples of physical states in $\Omega$, as depicted in Figure~\ref{cone with section}. The \emph{dimension} $d$ of the GPT is by definition the dimension of $V$ as a vector space, i.e.\ $d\coloneqq \dim V = \dim \Omega +1$.  The `normalising' functional $u$, also called \emph{order unit}, satisfies $u(\omega)\equiv 1$ for all $\omega\in \Omega$. The cone $C\coloneqq \{\lambda \omega:\, \lambda\geq 0,\, \omega\in \Omega\}$ is called the \emph{cone of (unnormalised) states}. We will usually assume that $C$ enjoys some basic properties that correspond to our intuitive notion of a well-behaved cone, which we signify by calling it \emph{proper} (see the Methods section for a rigorous definition).
Observe that $\Omega=C\cap u^{-1}(1)$. Mathematically, this gives $V$ the structure of an \emph{ordered vector space}: the ordering is defined for $x,y\in V$ by saying that $x\leq y$ if $y-x\in C$. Observe that, unlike that between real numbers, this ordering is not \emph{total}, i.e.\ it is possible that neither $x\leq y$ nor $y\leq x$.

A physical theory needs measurements in addition to states. The probability that a certain outcome of a given measurement occurs is a function of the state, called an \emph{effect} and denoted by $e:\Omega\to [0,1]$. It can be shown that in order to preserve our interpretation of stochastic state preparations $e$ must in fact be convex-linear, which can be expressed mathematically by requiring that $e(p\omega + (1-p) \tau) = p\,e(\omega) + (1-p) e(\tau)$ for all $\omega,\tau\in\Omega$ and all probabilities $p\in [0,1]$. It is then possible to make $e$ linear by extending its action from $\Omega$ to the whole $V$. Linear functionals acting on $V$ form themselves a vector space, called the \emph{dual} of $V$ and denoted with $V^*$. It is possible to make $V^*$ an ordered vector space in a natural way: for $f,g\in V^*$, we say that $f\leq g$ if $f(\omega)\leq g(\omega)$ for all $\omega\in \Omega$ (equivalently, for all $\omega\in C$). Positive functionals form again a cone, called the \emph{dual cone} to $C$ and denoted with $C^*$.
The operational requirement that each effect should produce a probability when evaluated on a physical state can then be very naturally rephrased as the two-fold inequality $0\leq e\leq u$, to be understood as holding with respect to the ordering of $V^*$. Hence, in the GPT formalism a \emph{measurement} is a (finite) collection of effects $(e_i)_{i\in I}$, where each $e_i\in V^*$ satisfies $e_i\geq 0$, and the completeness condition for the outcome probabilities further imposes that $\sum_{i\in I} e_i = u$. The probability of obtaining the outcome $i$ when measuring the state $\omega$ is given by $e_i(\omega)$.

It is a separate question, related to the physics of the system under consideration, whether \emph{any} collection of effects with the above properties may be implemented as an actual physical measurement. This assumption is usually called the \emph{no-restriction hypothesis}~\cite{no-restriction}. As classical and quantum mechanics both satisfy it, we will henceforth include it in our list of assumptions. In fact, it will play an important role in deriving the implications of our results for the foundations of physics.
With the no-restriction hypothesis, the list of rules to translate physical experiments into the mathematical formalism -- and vice versa -- and to make probabilistic predictions about their results is complete. Since all we need is the triple $A=(V,C,u)$, where $V$ is the host vector space, $C$ the cone of states, and $u$ the order unit, we will often identify GPTs with said triples.

\begin{figure}[htbp]
\begin{tikzpicture}[scale=1]
\draw[fill=gray!30,thick] (0,3) circle [x radius=1.5cm, y radius=0.45cm];
\coordinate (origin) at (0,0);
\filldraw (0,0) circle (2pt);
\node at (0,-0.33) {\large $0$};
\draw[thick] (origin) -- (-1.491,2.95);
\draw[thick] (origin) -- (1.491,2.95);
\draw[thick,dotted] (1.491,2.95) -- (2.2365, 4.425);
\draw[thick,dotted] (-1.491,2.95) -- (-2.19177, 4.3365);
\draw (-4.2,2) -- (1.7,2) -- (4.2,4) -- (-1.7,4) -- (-4.2,2);
\node (Omega) at (0.7,3) {\large $\Omega$};
\node (C) at (1,0.9) {\large $C$};
\node (u) at (-2.4,2.5) {\large $u=1$};
\node (V) at (-3.5,3.5) {\large $V$};
\end{tikzpicture}
\caption{The basic ingredients of a GPT are a real finite-dimensional vector space $V$ and a cone $C$. The order unit functional $u$ defines a hyperplane $u^{-1}(1)$, whose intersection with $C$ identifies the state space $\Omega$.}
\label{cone with section}
\end{figure}
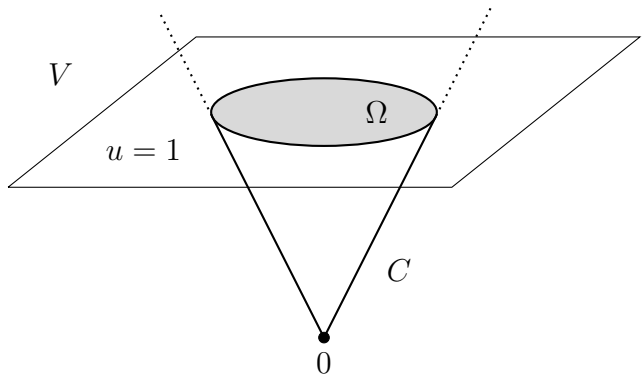

\begin{figure}[htbp]
\begin{tikzpicture}[scale=1]
\coordinate (origin) at (0,0);
\coordinate (u) at (0,4);
\filldraw (0,4) circle (2pt);
\filldraw (0,0) circle (2pt);
\node at (0,-0.33) {\large $0$};
\node at (0,4.33) {\large $u$};
\draw[fill=gray!30,thick] (0,2) circle [x radius=3cm, y radius=0.45cm];
\draw[thick] (origin) -- (2.945,1.913);
\draw[thick] (origin) -- (-2.945,1.913);
\draw[thick] (u) -- (2.945,2.087);
\draw[thick] (u) -- (-2.945,2.087);
\draw[thick,dotted] (2.945,1.913) -- (4.32915, 2.81211);
\draw[thick,dotted] (-2.945,1.913) -- (-4.32915, 2.81211);
\node at (0,3) {\large effects};
\node at (4,2) {\large $C^*$};
\node (V*) at (-3.5,3.5) {\large $V^*$};
\end{tikzpicture}
\caption{The dual space $V^*$, comprising the dual cone $C^*$, the order unit $u$, and the set of effects defined by the order interval $[0,u]=\{e\in V^*:\, 0\leq e\leq u\}$.}
\label{dual cone}
\end{figure}
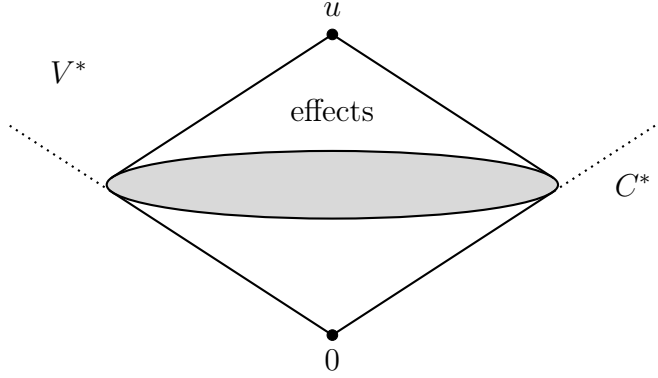



\subsection{Universal definition of entanglement} \label{subsec universal entanglement}

Given that our GPT machinery is supposed to account for real physics, it is very natural to wonder what it tells us as far as multipartite systems are concerned. Namely, given two physical systems represented as GPTs $A=(V_1, C_1, u_1)$ and $B=(V_2, C_2, u_2)$, is there a way to represent also the joint system as a GPT $AB=\left(V_{12}, C_{12}, u_{12}\right)$?
Under some natural assumptions on physical composites, it can be shown that vector spaces and order units obey a simple tensor product rule~\cite{tensor-rule-1, tensor-rule-2}:
\begin{equation}
V_{12} = V_1\otimes V_2\, ,\quad u_{12}=u_1\otimes u_2\, .
\label{tensor rule space}
\end{equation}
The main hidden hypothesis that leads to Eq.~\eqref{tensor rule space} is the \emph{local tomography principle}, i.e.\ the assumption that any state of the joint system is completely determined by the statistics it yields under local measurements.

The status of the cone $C_{12}$ is far more delicate. Some natural constraints come from the requirement that convex combinations of product states be allowed as legitimate states of the joint system, and -- dually -- that local measurement be allowed as legitimate measurements on the joint system. These considerations lead to the two-fold bound
\begin{equation}
C_1 \tmin C_2 \subseteq C_{12} \subseteq C_1 \tmax C_2\, .
\label{double bound}
\end{equation}
Here, $C_1 \tmin C_2$ is called the \emph{minimal tensor product}, and contains product states and convex combinations thereof. The \emph{maximal tensor product} $C_1 \tmax C_2$, instead, includes all those tensors that are positive on product effects~\cite{Peressini-minmax, Hulanicki-minmax}. In formulae,
\begin{align}
C_1\tmin C_2 &\coloneqq \conv \left\{ x\otimes y:\ x\in C_1,\, y\in C_2\right\} , \label{C min} \\
C_1\tmax C_2 &\coloneqq \left\{ z\in V_1\otimes V_2:\ (f\otimes g)(z)\geq 0\ \forall\, f\in C_1^*,\, g\in C_2^*\right\} , \label{C max}
\end{align}
where $\conv$ denotes the convex hull.

All cones satisfying~\eqref{double bound} should be regarded as identifying a priori valid composition rules. As there appears to be no general and indisputable physical principle that is capable of singling out a special one among them, a sensible decision requires a deeper investigation of the physics of the system under examination. In some sense, a composite is more than the sum of its parts.

\begin{example}[Classical theories as GPTs] \label{example classical}
At this point it is instructive to discuss some notable examples of GPTs. A classical theory has by definition a fixed (finite) number $d$ of perfectly distinguishable configurations $\omega_1,\ldots, \omega_d$, and its state is described by the probabilities of it being in each of those configurations. This means that the state space $\Omega$ has the geometric shape of a \emph{simplex}, and that the corresponding cone $C$, called a \emph{classical cone}, has the form $C=\left\{\sumno_i \lambda_i \omega_i:\, \lambda_i\geq 0\ \forall\, i\right\}$, with $\{\omega_i\}_i$ being a vector basis of the space $V$. The order functional $u$ encodes the normalisation, and it acts as $u\left(\sum_i x_i \omega_i\right) = \sum_i x_i$.

Classical theories are special as far as composites are concerned. In fact, it can be shown that if $C$ is a classical cone, then for \emph{all} other cones $C'$ one has that $C \tmin C' = C \tmax C'$, so that the chain of inclusions in Eq.~\eqref{double bound} collapses and leads to no ambiguity.
\end{example}

\begin{example}[Quantum mechanics as a GPT] \label{example quantum}
An $n$-level quantum system can also be described by means of the GPT formalism. 
In this case, the host vector space $V$ is simply the real vector space $\Herm_n$ of $n\times n$ Hermitian matrices. Unnormalised states form the cone $\PSD_n$ of positive semidefinite matrices. Since physical states, a.k.a.\ density matrices, are obtained by further normalising the trace to $1$, we see that the order unit $u$ is nothing but the trace functional. We are thus left with the GPT $\mathrm{QM}_n \coloneqq \left( \Herm_n,\, \PSD_n,\, \tr\right)$.

The quantum composition rule is easily expressed in words: the bipartite system obtained from an $n$-level and an $m$-level quantum system is simply an $nm$-level quantum system. This definition is easily seen to satisfy Eq.~\eqref{tensor rule space}, so we move on to Eq.~\eqref{double bound}. States in the minimal tensor product are precisely those that do not exhibit \emph{entanglement}, a.k.a.\ \emph{separable states}~\cite{Werner}. At the opposite end of the spectrum, Hermitian operators in the maximal tensor product are known as \emph{entanglement witnesses}~\cite{HorodeckiPPT, Terhal2000}. As is well known, the cone $\PSD_{nm}$ of physical quantum states pertaining to the bipartite system is neither of the two: it includes entangled states that are not in $\PSD_n \tmin \PSD_m$, such as the maximally entangled state, yet it leaves out certain non-positive matrices that lie inside $\PSD_n\tmax \PSD_m$, e.g.\ the flip operator $F$~\cite{Werner}. Therefore, quantum theory demonstrates that a composition rule that makes both inclusions in Eq.~\eqref{double bound} strict may be the one prescribed by Nature.
\end{example}

By analogy with the quantum concept, for any two cones $C_1,C_2$ we call elements of $C_1\tmin C_2$ \emph{separable} and elements of $C_1 \tmax C_2$ that are not in $C_1\tmin C_2$ \emph{entangled}. Exactly as in quantum mechanics, separable states of a bipartite GPT can be prepared with local operations and shared randomness on separated systems. Observe that not all entangled elements of $C_1 \tmax C_2$ necessarily represent legitimate states in the GPT interpretation, as it appears from Eq.~\eqref{double bound}. However, if it holds that
\begin{equation}
C_1 \tmin C_2 \neq C_1 \tmax C_2\, ,
\label{min neq max}
\end{equation}
then \emph{the corresponding GPTs must exhibit entanglement, either at the level of states or at the level of measurements}. In fact, Eq.~\eqref{min neq max} implies that every physically allowed cone $C_{12}$ satisfying Eq.~\eqref{double bound} is such that either $C_{12}\supsetneq C_1\tmin C_2$, i.e.\ there are entangled states, or $C_{12}^*\supsetneq C_1^*\tmin C_2^*$, i.e.\ there are entangled measurements. A pair of cones $(C_1,C_2)$ is called \emph{entangleable} if it satisfies Eq.~\eqref{min neq max}. This notion of entangleability of physical theories as modelled by GPTs -- or, more generally, of entangleability of cones -- plays a central role in our work.



\subsection{Main question and findings}

The fundamental question we investigate here concerns the connection between the notion of classicality defined in Example~\ref{example classical} and the above notion of entangleability. The nature of such a connection is apparent in the quantum mechanical formalism: the possibility of constructing superpositions of states at the single-system level, which is a manifestation of non-classicality, directly implies the existence of entangled states in bipartite systems. The problem we pose here is whether this implication is just an accident of quantum theory, or on the contrary it is a universal feature of the general logical rules governing composition of physical theories. Thanks to the discussion at the end of Subsection~\ref{subsec universal entanglement}, we can formulate the question in precise mathematical terms:
\begin{center} \it
Which pairs of general probabilistic theories are entangleable? Mathematically, can we characterise all entangleable pairs of proper cones?
\end{center}

The following easily established fact was mentioned above: \emph{if either $C_1$ or $C_2$ is a classical cone, then the pair $(C_1,C_2)$ is not entangleable}~\cite{NP}. This corresponds to the rather intuitive statement that classical systems cannot become entangled with any other system. A partial converse to this was proved long ago by Namioka and Phelps: \emph{if a cone $C_1$ is such that $(C_1,C_2)$ is not entangleable for every other cone $C_2$, then $C_1$ must be classical}~\cite{NP}. This latter result is conceptually important because it provides a partial answer to the above question. However, it does not allow us to conclude anything for a \emph{single} pair of theories, which is arguably the most significant case if one wants to establish universality of entanglement as a physical phenomenon. 

Our main results answer the above question for a wide class of GPTs that encompasses most physically reasonable models. We start by looking at the simplest case of all, to wit, that of two $3$-dimensional GPTs. Besides being interesting on its own, its solution will turn out to be critical to the understanding of more general cases. 

\begin{result} \label{3-dim result}
Let $(C_1,C_2)$ be two $3$-dimensional cones. Then $(C_1,C_2)$ is entangleable if and only if neither $C_1$ nor $C_2$ is classical.
\end{result}

It has been proposed that due to the allegedly discrete nature of space-time~\cite{Feynman1982}, physical state spaces may themselves be ultimately discrete~\cite{Bunyi2005}, meaning that the number of accessible pure states in a finite-dimensional system may be finite. This would offer some advantages on the interpretational side~\cite{Bunyi2006}, although it would require modification of the post-measurement collapse rule~\cite{Pfister-no-disturbance}. These speculations motivate us to answer the above question for the special case of discrete state spaces. From the mathematical standpoint, a state space hosting only a finite number of pure states is modelled by a convex set with only finitely many extreme points, i.e.\ a polytope. The corresponding cone of states will then be a \emph{polyhedral cone}. Our next result then reads as follows.

\begin{result} \label{polyhedral result}
Let $(C_1,C_2)$ be two proper polyhedral cones. Then $(C_1,C_2)$ is entangleable if and only if neither $C_1$ nor $C_2$ is classical.
\end{result}

We are also able to tackle the important special case of one of the two GPTs being quantum theory. This problem has been considered before in~\cite{Fritz2017, Huber-Netzer, Passer2018}, with an entirely different motivation. Our next result improves upon~\cite[Theorem~4.1]{Passer2018}, answering our main question in yet another case.

\begin{result} \label{semiquantum result}
Let $\PSD_n$ be the cone of $n\times n$ positive semidefinite matrices. For a proper cone $C$ in dimension $d$, the pair $(C, \PSD_n)$ is entangleable if and only if $C$ is not classical, provided that $\floor{\log_2 n} \geq \frac{d-1}{2}$.
\end{result}

From Results~\ref{3-dim result}--\ref{semiquantum result} it is apparent that local non-classicality is intimately connected with global entangleability, as we discuss more thoroughly below. So far, we have explored this connection in a fundamentally qualitative way. However, it is also possible to ask a quantitative version of our main question: \emph{given any two GPTs that are non-classical to some quantifiable extent, can we estimate their degree of entangleability, that is, the maximum possible entanglement exhibited by global states?}
Note that the answer will in general depend on the measures we employ to gauge the global entanglement and the local non-classicality. Although it is significantly more complex than its qualitative counterpart, a solution to this problem can nevertheless be found for all those theories -- called \emph{symmetric} -- whose state space is centrally symmetric with respect to some centre. Classical theories of dimension $d>2$ are automatically ruled out by this assumption, which makes the problem more tractable. However, in spite of its geometric appeal, central symmetry is perhaps not a natural requirement from a physical perspective, as e.g.\ no quantum system besides that of a single qubit has a symmetric state space. Yet, it is remarkable that a 
complete solution can be found for such a general class of examples.

\begin{result} \label{symmetric result}
Given any pair of symmetric GPTs of dimensions $n+1,\, m+1\geq 3$, their maximal tensor product contains a state whose entanglement robustness~\cite{VidalTarrach} is at least $\erob (n,m) \geq (r(n,m)-1)/2$, where $r(n,m)$ is the universal function called `projective/injective ratio' and defined in~\cite[Eq.~(15)]{XOR}. In particular, $\erob(n,m)\geq 1/36$ for all $n,m\geq 2$,  and asymptotically $\erob (n,m)\geq c \min\{n,m\}^{1/8-o(1)}$ for some constant $c>0$. Hence, all pairs of symmetric GPTs are entangleable, with the maximal robustness of entanglement growing unboundedly with the minimum local dimension.
\end{result}

Our conceptual contributions extend far beyond providing an answer to our main question in several physically interesting cases, which marks in itself some tangible progress in a long-standing open problem. In fact, owing to their versatility and generality, the techniques we develop constitute per se a significant step forward, both conceptually and mathematically. These techniques include e.g.\ an innovative use of the order-theoretic concept of retract, which allows us to study constrained GPTs, a general framework to construct and detect general entangled states in a bipartite GPT, and a systematic connection with a recently developed functional-analytic theory of tensor norm ratios~\cite{XOR}.





\section{Discussion}

Our Results~\ref{3-dim result}--\ref{symmetric result} demonstrate that there is a profound connection between the notion of non-classicality and that of entanglement. While it was long known that the former is a necessary condition for the latter, we have proved that the two are actually equivalent for a large class of cases of immediate interest for the foundations of physics. We have shown that this connection goes far beyond quantum mechanics, and characterises instead all theories that can be modelled within the GPT formalism. Let us remark in passing that this type of model-independent approach to the study of operational features of physical theories has a long history~\cite{PR-boxes, Barnum-no-broad, Barrett-original, Barnum2008, telep-in-GPT, ultimate, Sikora2017}.
For the case study of symmetric GPTs, we have been able to make the aforementioned connection quantitative. What our results suggest is that, in a bipartite system whose local components exhibit some non-classical behaviour, entanglement of states or measurements becomes logically unavoidable. We can conjecture that this is a fully general behaviour: \emph{all pairs of non-classical GPTs may be entangleable.} 

The mathematical translation of this conjecture is that \emph{all pairs of non-classical proper cones may be entangleable.} Interestingly, this same problem was formulated long ago by Barker, with an entirely different and purely mathematical motivation, and has been open since~\cite{Barker1976, Barker-review}. We have provided the first convincing evidence that the above conjecture may be true in general, proving it in the first nontrivial case of dimension $3$ (Result~\ref{3-dim result}) and for polyhedral cones (Result~\ref{polyhedral result}). Remarkably, this question is already implicit in previous work by Namioka and Phelps~\cite{NP}, of which Barker seem to have been unaware. The same sort of problem was again rediscovered more recently, in a somewhat limited setting in which one of the two theories is set to be quantum mechanics~\cite{Fritz2017, Huber-Netzer, Passer2018}. There, the motivation is again entirely different, coming from operator system theory. Once reformulated in our language, the results from~\cite{Fritz2017, Huber-Netzer, Passer2018} state that: (a) for a given polyhedral cone $C$, the pair $(C, \PSD_n)$ is entangleable if and only if $C$ is non-classical; (b) for any cone $C$ in a $d$-dimensional space, provided that $\log_2 n\geq d-2$, it holds that $(C,\PSD_n)$ is entangleable if and only if $C$ is non-classical. Our techniques lead to a more direct proof of (a), as well as showing that (b) holds under the weaker condition $\floor{\log_2 n} \geq (d-1)/2$ (Result~\ref{semiquantum result}).

From the mathematical standpoint, our main question 
connects very naturally to the problem of evaluating the minimal constant of domination of the injective over the projective tensor norm for Banach spaces of fixed local dimensions~\cite{XOR}. In fact, given any finite-dimensional Banach space, we can construct a GPT by declaring its unit ball to be our state space~\cite[\S~2.3.3]{lamiatesi}. The entanglement robustness of certain bipartite states can then be expressed by means of the ratio between projective and injective tensor norm of the corresponding tensors. Consequently, the maximum entanglement robustness in a bipartite system is directly linked to the constant of domination of the latter over the former norm, which can be estimated using the techniques of~\cite{XOR} (Result~\ref{symmetric result}).

Result~\ref{symmetric result} is affected by the geometrically natural yet physically questionable restriction to symmetric models, and should therefore be regarded more as the starting point of a quantitative investigation of our main question.
Ultimately, we envision the existence of \emph{a general lower bound on the minimal amount of entanglement in the maximal tensor product of two GPTs in terms of their local non-classicality.} To prove such a statement one would need to construct: (i)~a suitable measure of non-classicality, i.e.\ a functional $\nu$ that assigns to every GPT $A=(V,C,u)$ a non-negative real number $\nu(A)$, in such a way that $\nu(A)=0$ if and only if $C$ is a classical cone; and (ii)~a general measure of entanglement for bipartite states in GPTs, i.e.\ a function $E$ that, given two local GPTs $A,B$, assigns a non-negative real number $E(\omega_{AB})$ to every state $\omega_{AB}$ in the state space $\Omega_{A\tmax B}$ corresponding to the maximal tensor product of the cones, in such a way that $E(\omega_{AB})=0$ if and only if $\omega_{AB}\in \Omega_{A\tmin B}$ is a separable state. Within this framework, a quantitative relation between local non-classicality and global entanglement would read
\begin{equation}
    \max_{\omega_{AB}\,\in\, \Omega_{A\tmax B}} E(\omega_{AB}) \geq F(\nu (A), \nu(B))\, ,
    \label{quantitative entangleability}
\end{equation}
where $F:\R_+\times \R_+\to \R_+$ is some universal function with the property that $F(x,y)=0$ only when either $x$ or $y$ equals $0$. In the statement of Result~\ref{symmetric result} we chose as $E$ the entanglement robustness~\cite{VidalTarrach, Takagi2019}. We believe that of all entanglement measures constructed in the field of quantum information~\cite{HAYASHI}, the entanglement robustness stands out as a natural candidate to appear in Eq.~\eqref{quantitative entangleability}, as it relies only on the convex character of the theory, and as such it carries over swiftly to the GPT formalism. On the contrary, we do not yet have such a clear ansatz for the non-classicality measure $\nu$.


The problem we study here admits many possible variations. For instance, a stronger question to ask would be whether in any pair of non-classical GPTs one can violate a Bell inequality~\cite[Definition~2.14]{lamiatesi}. Since Bell inequalities can be violated only by entangled states, this would immediately imply that the two GPTs are entangleable. Solving such a problem would lead us to the stronger conclusion that the entanglement exhibited by non-classical theories can also be experimentally accessed in the form of some non-local correlations, thus enabling device-independent information theory in GPTs~\cite{Brunner-review}. A quantitative answer to this question would translate to an inequality analogous to Eq.~\eqref{quantitative entangleability}, with a measure of non-locality such as the maximal violation of a CHSH-type inequality~\cite{CHSH} or other more general measures (see e.g.~\cite{Pal2014}) displacing the entanglement measure $E$. For some partial results in this direction, see e.g.~\cite[Theorem~2.39]{lamiatesi}, which builds upon previous works~\cite{Wolf-incompatible, Banik-2013, Jencova2017}. 




\section{Methods}


In this section we discuss the proof ideas of Results~\ref{3-dim result}--\ref{symmetric result}. For a complete presentation with all the technical details, we refer the reader to the Supplementary Information.

\subsection{Technical background}
As discussed above, a GPT is a triple $(V,C,u)$, where $V$ is a finite-dimensional real vector space, $C\subset V$ is a proper cone, and $u\in \inter(C^*)$ is a strictly positive functional on $C$. A subset $C\subseteq V$ is called a cone if $\lambda C=C$ for all $\lambda>0$; a cone is said to be salient if $C\cap (-C)=\{0\}$, generating if $C-C=\mathspan(C)=V$, and proper if it is convex, topologically closed, salient, and generating. A classical cone $C$ is one that is generated by a basis of $V$, i.e.\ $C=\left\{\sumno_i \lambda_i e_i:\, \lambda_i\geq 0\ \forall\, i\right\}$ for some basis $\{e_i\}_i$. Any proper cone $C\subset V$ can be declared to be the set of positive vectors of $V$ and thus induces an ordering on $V$. Positive functionals in the dual space $V^*$ form the dual cone $C^*$. If $C$ is proper, then $C^{**}=C$ modulo the identification $V^{**}=V$.

Two cones $C_1,C_2$ can be combined according to either the minimal or the maximal tensor product, defined in Eq.~\eqref{C min} and~\eqref{C max}, respectively. Observe that $C_1\tmin C_2\subseteq C_1\tmax C_2$, because products of positive functionals take on positive values when evaluated on products of positive vectors. Remember that we call the pair $(C_1,C_2)$ entangleable if this inclusion is strict. The following easily verified and well-known fact~\cite{Barker1976, Barker-review} is a cornerstone of our intuition concerning these products, so we present a proof for the benefit of the reader.

\begin{lemma} \label{lemma:easy-direction}
Let $C_1,C_2$ be proper cones, at least one of which is classical. Then $(C_1,C_2)$ is not entangleable.
\end{lemma}

\begin{proof}
Assume that $C_1$ is generated by a basis $\{e_i\}_i$ of $V_1$, and consider the dual basis $\{e_i^*\}_i$ of $V_1^*$, which satisfies $e_i^*(e_j)=\delta_{i,j}$. Decompose an arbitrary $z \in C_1 \tmax C_2$ as $z = \sumno_i e_i \otimes x_i$, where $x_i\in V_2$. By definition of maximal tensor product, for every $f \in C_2^*$ we have that $0 \leq (e_i^* \otimes f)(z) = f(x_i)$ for all $i$. This shows that $x_i \in C_2^{**} = C_2$, where we used the fact that $C_2$ is proper. Hence, $z \in C_1 \tmin C_2$, and consequently $C_1\tmax C_2 = C_1 \tmin C_2$.
\end{proof}

\subsection{Proof of Result~\ref{3-dim result}}

Here we shall prove that any pair $(C_1,C_2)$ of non-classical $3$-dimensional proper cones is entangleable. Our strategy can be summarised as follows.
\begin{enumerate}[(i)]
\item We will apply linear isomorphisms $\Phi_i$ to bring both cones $C_i$ into `standard' forms, for which we can find simpler cones $C_i', C_i''$ such that $C_i'\subseteq \Phi_i(C_i)\subseteq C_i''$. Note that $(C_1,C_2)$ is entangleable if and only if $(\Phi(C_1),\Phi(C_2))$ is such.
\item We will then lower bound $C_1\tmax C_2\supseteq C_1'\tmax C_2'$, and upper bound $C_1\tmin C_2\subseteq C_1''\tmin C_2''$. Assuming by contradiction that $C_1\tmax C_2\subseteq C_1\tmin C_2$, it follows that $C_1'\tmax C_2'\subseteq C_1''\tmin C_2''$.
\item However, using the relatively simple structure of $C_i',C_i''$, we will explicitly show that $C_1'\tmax C_2'\subsetneq C_1''\tmin C_2''$.
\end{enumerate}

A natural way to construct a cone is through one of its sections. Namely, given a convex set $K\subseteq V$, let us define the cone
\begin{equation} \label{eq:cone-over-K}
\CC(K) = \{ (tx,t) \st x \in K, \ t \in \R_+ \} \subseteq V\times \R\, .
\end{equation}
One can verify that $\CC(K)$ is a proper cone if and only if $K\subset V$ is a convex body (compact convex set with non-empty interior), and that it is non-classical if and only if $K$ is not a simplex. Moreover, every proper cone in dimension $d$ is linearly isomorphic to $\CC(K)$ for some $(d-1)$-dimensional convex body $K$.\footnote{We already knew this from the GPT setting: we used proper cones in the definition of a GPT precisely because they admit suitable sections -- namely, state spaces.}
In our case, to generate $3$-dimensional cones we need to look at $2$-dimensional convex bodies $K$, which allows us to use good old planar geometry to tackle the problem. We start by defining two special convex sets: the \emph{kite} with centre $(a,b)$ (where $-1<a,b<1$) is constructed as 
\begin{equation}
    T_{a,b} \coloneqq \conv \{ (a,\pm 1), (\pm 1,b) \};
    \label{kite}
\end{equation}
the \emph{blunt square}, instead, is simply the unit square without its corners: 
\begin{equation}
    S = [-1,1]^2 \setminus \{-1,1\}^2.
    \label{blunt square}
\end{equation}
For a pictorial representation of these two sets, see Figure~\ref{kite-blunt-square}.

\begin{figure}[htbp] \begin{center}
\begin{tikzpicture}[scale=2]
	\coordinate (a0) at (1,1);
	\coordinate (a1) at (0.95,1);
	\coordinate (a2) at (1,0.95);
	\coordinate (b0) at (1,-1);
	\coordinate (b1) at (0.95,-1);
	\coordinate (b2) at (1,-0.95);
	\coordinate (c0) at (-1,-1);
	\coordinate (c1) at (-0.95,-1);
	\coordinate (c2) at (-1,-0.95);
	\coordinate (d0) at (-1,1) ;
	\coordinate (d1) at (-0.95,1) ;
	\coordinate (d2) at (-1,0.95);
	\coordinate (A) at (0.6,1);
	\coordinate (B) at (1,0.2);
	\coordinate (C) at (0.6,-1);
	\coordinate (D) at (-1,0.2) ;
	\draw (A)--(B)--(C)--(D)--(A);
	\draw (a1)--(d1);
	\draw (a2)--(b2);
	\draw (b1)--(c1);
	\draw (c2)--(d2);
	\draw (A) node {$\bullet$};
	\draw (A) node[above] {$(a,1)$};
    \draw (B) node {$\bullet$};
	\draw (B) node[right] {$(1,b)$};
    \draw (C) node {$\bullet$};
	\draw (C) node[below] {$(a,-1)$};
    \draw (D) node {$\bullet$};
	\draw (D) node[left] {$(-1,b)$};
	\draw (a0) node[right] {$(1,1)$};
	\draw (b0) node[right] {$(1,-1)$};
	\draw (c0) node[left] {$(-1,-1)$};
	\draw (d0) node[left] {$(-1,1)$};
	\draw (0.2,0.4) node[left] {\large $T_{a,b}$};
	\draw (-0.3,0.75) node[left] {\large $S$};
    \end{tikzpicture} \end{center}
\caption{The kite $T_{a,b}$ and the blunt square $S$, defined in Eq.~\eqref{kite} and~\eqref{blunt square}, respectively.} \label{kite-blunt-square}
\end{figure}
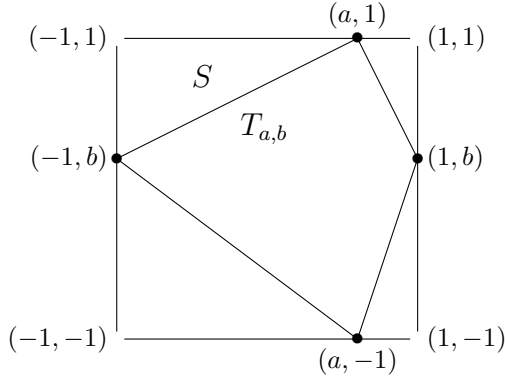

The reason why we are interested in kites and blunt squares is that, apart from triangles, any $2$-dimensional convex set can be inscribed between one and the other by the application of a suitable linear isomorphism. This analogue of Auerbach's lemma~\cite[Vol~I, \S~1.c.3]{LINDENSTRAUSS} for $2$-dimensional convex bodies can be formalised as follows.

\begin{proposition} \label{proposition:auerbach}
Let $V$ be a $3$-dimensional vector space, and $C \subset V$ a proper cone which is not classical. 
There exist $(a,b) \in (-1,1)^2$ and a linear bijection $\Phi : V \to \R^2 \times \R$ such that
\[ \CC(T_{a,b}) \subseteq \Phi(C) \subseteq \CC(S) .\]
\end{proposition}

We refer to the Supplementary Information for a proof. By the discussion at the beginning of the section, it should be clear that a statement such as Proposition~\ref{proposition:auerbach} allows to focus our effort on the pairs of lower and upper bounds rather than on the original cones. The main technical contribution of this section completes the analysis by studying the properties of minimal and maximal tensor products of cones generated by kites and blunt squares.

\begin{proposition} \label{proposition:sticks-out}
Take four numbers $-1<a_1,a_2,b_1,b_2<1$. Then
\begin{equation}
\CC(T_{a_1,b_1}) \tmax \CC(T_{a_2,b_2}) \not\subseteq \CC(S) \tmin \CC(S).
\end{equation}
In other words, there exists $\omega \in \CC(T_{a_1,b_1}) \tmax \CC(T_{a_2,b_2})$ such that $\omega \not\in \CC(S) \tmin \CC(S)$.
\end{proposition}

The proof of Proposition~\ref{proposition:sticks-out} is constructive: we exhibit an explicit tensor $\omega\in \R^3\otimes \R^3$ and show that it belongs to $\CC(T_{a_1,b_1}) \tmax \CC(T_{a_2,b_2})$ but not to $ \CC(S) \tmin \CC(S)$. The former fact can be proved by a direct computation. For the latter, instead, we construct a Bell-type expression that is strictly less than $2$ on the whole $\CC(S) \tmin \CC(S)$, yet it evaluates precisely to $2$ on $\omega$. With these tools at hands, we are now in position to prove our first main result.

\begin{proof}[Proof of Result~\ref{3-dim result}]
Considering two $3$-dimensional non-classical proper cones $C_1$, $C_2$, we show by contradiction that the pair $(C_1,C_2)$ is entangleable. Up to the application of local linear isomorphisms on $C_1,C_2$, and using Proposition~\ref{proposition:auerbach}, we may assume that
\begin{equation}
\CC(T_{a_1,b_1}) \subseteq C_1 \subseteq \CC(S) \textnormal{ and } \CC(T_{a_2,b_2}) \subseteq C_2 \subseteq \CC(S)
\end{equation}
for some numbers $-1<a_1,a_2,b_1,b_2<1$. Since $\tmin$ and $\tmax$ are increasing operations with respect to set inclusion, it follows that
\begin{equation}
    \CC(T_{a_1,b_1}) \tmax \CC(T_{a_2,b_2}) \subseteq C_1 \tmax C_2 = C_1 \tmin C_2 \subseteq \CC(S) \tmin \CC(S),
\end{equation}
which contradicts the conclusion of Proposition~\ref{proposition:sticks-out}.
\end{proof}

\subsection{Proof of Result~\ref{polyhedral result}}

Throughout this section we will prove that all pairs of non-classical polyhedral cones are entangleable. We call a cone $C$ \emph{polyhedral} if there are finitely many vectors $\{v_i\}_i$ such that $C = \left\{\sumno_i \lambda_i v_i:\, \lambda_i\geq 0\ \forall\ i \right\}$. One of the main tools we employ here is the concept of retract. Given two vector spaces $V,V'$ ordered by proper cones $C,C'$, a linear map $\Phi:V\to V'$ is called \emph{positive} if $\Phi(C)\subseteq C'$. We then say that $C'$ is a \emph{retract} of $C$ if there are positive maps $\Phi:V\to V'$ and $\Psi:V'\to V$ such that $\Phi \circ \Psi = \Id_{V'}$. When this happens, $C'$ can be seen as a sub-cone of $C$ that is also the image of a positive projection. For example, we shall see that facets of polyhedral cones are always retracts. 

To appreciate the importance of retracts for the study of entangleability, we first need to familiarise ourselves with the transformation properties of minimal and maximal tensor products under local positive maps. Consider cones $C_i\subseteq V_i$ and positive maps $\Phi_i:V_i\to V'_i$ ($i=1,2$). Then
\begin{align}
    \left( \Phi_1\otimes \Phi_2\right) \left( C_1\tmin C_2 \right) &\subseteq \Phi_1(C_1) \tmin \Phi_2(C_2)\, , \label{positive action tmin} \\
    \left( \Phi_1\otimes \Phi_2\right) \left( C_1\tmax C_2 \right) &\subseteq \Phi_1(C_1) \tmax \Phi_2(C_2)\, . \label{positive action tmax}
\end{align}
The former inclusion can be verified directly by means of the decomposition of tensors in $C_1\tmin C_2$. As for the latter, take $z\in C_1\tmax C_2$ and a pair of functionals $f_i$ such that $f_i(\Phi_i(x_i))\geq 0$ for all $x_i\in C_i$. Using the concept of adjoint map,\footnote{The adjoint of a linear map $\Phi:V\to W$ is the linear map $\Phi^*:W^*\to V^*$, where $V^*,W^*$ are the dual spaces to $V,W$, uniquely defined by $(\Phi^*f)(x)\equiv f\left( \Phi(x)\right)$, for all $x\in V$ and $f\in W^*$.} we can express this condition as $\Phi_i^*(f_i)\in C_i^*$, where $C_i^*$ is the dual cone to $C_i$. Hence,
\begin{equation*}
    (f_1\otimes f_2)\left(\left( \Phi_1\otimes \Phi_2\right)(z)\right) = \left(\Phi_1^*(f_1)\otimes \Phi_2^*(f_2)\right)(z) \geq 0\, ,
\end{equation*}
which proves Eq.~\eqref{positive action tmax}.

\begin{proposition}[Entangleability from retracts] \label{proposition:retracts-nuclear}
For $i=1,2$, let $C_i\subset V_i$ be proper cones with retracts $C'_i\subset V'_i$. If $(C'_1,C'_2)$ is entangleable, then so is $(C_1,C_2)$.
\end{proposition}

\begin{proof}
Let us denote by $\Phi_i:V_i\to V'_i$ and $\Psi_i:V'_i\to V_i$ the linear maps associated to the corresponding retracts. Assume that $(C_1,C_2)$ is not entangleable, so that $C_1\tmin C_2=C_1\tmax C_2$. Then
\begin{align*}
    C'_1\tmax C'_2 &= \left((\Phi_1\circ \Psi_1) \otimes (\Phi_2\circ \Psi_2)\right) \left( C'_1\tmax C'_2\right) \\
    &\textsubseteq{(i)} (\Phi_1 \otimes \Phi_2) \left( \Psi_1(C'_1) \tmax \Psi_2(C'_2) \right) \\
    &\textsubseteq{(ii)} (\Phi_1 \otimes \Phi_2) \left( C_1 \tmax C_2 \right) \\
    &\texteq{(iii)} (\Phi_1 \otimes \Phi_2) \left( C_1 \tmin C_2 \right) \\
    &\textsubseteq{(iv)} \Phi_1(C_1) \tmin \Phi_2(C_2) \\
    &\textsubseteq{(v)} C'_1\tmin C'_2\, .
\end{align*}
Note that (i) comes from Eq.~\eqref{positive action tmax}, (ii) from the positivity of $\Psi_i$, (iii) from the unentangleability of $(C_1,C_2)$, (iv) from Eq.~\eqref{positive action tmin}, and finally (v) from the positivity of $\Phi_i$. Since we have shown that $C'_1\tmax C'_2\subseteq C'_1\tmin C'_2$ and the opposite inclusion is trivial, we conclude that $(C'_1,C'_2)$ is not entangleable.
\end{proof}

A possible strategy for demonstrating the entangleability of a pair of cones is then as follows: if we are able to exhibit two local retracts that are entangleable, then Proposition~\ref{proposition:retracts-nuclear} guarantees that so were the original cones.
In the case of polyhedral cones, the job of finding retracts is facilitated by the following lemma.

\begin{lemma} \label{lemma:facet-retract}
Let $F$ be a facet of a proper polyhedral cone $C$. Then $F$ is a retract of $C$.
\end{lemma}

The main idea of the proof of Lemma~\ref{lemma:facet-retract} is that it is always possible to `illuminate' a polyhedral cone with a collinear beam in such a way that its whole shadow lies inside one of its facets. 
The rigorous proof is relegated to the Supplementary Information. We now move on to the other main ingredient of the proof.

\begin{lemma} \label{lemma:simple-simplicial}
Let $C$ be a non-classical proper polyhedral cone with $\dim (C) \geq 4$. Then either $C$ or its dual $C^*$ has a facet which is non-classical.
\end{lemma}

Before we can apply Lemma~\ref{lemma:simple-simplicial} to our setting, we need to observe that \emph{retracts dualise.} This means that \emph{$C_1$ is a retract of $C_2$ if and only if $C_1^*$ is a retract of $C_2^*$,} for all pairs of proper cones $C_1,C_2$.

\begin{proof}[Proof of Result~\ref{polyhedral result}]
Let $C_1,C_2$ be non-classical proper polyhedral cones. Let us assume that e.g.\ $d_1\coloneqq \dim (C_1)\geq 4$, otherwise the claim follows from Result~\ref{3-dim result}. Thanks to Lemma~\ref{lemma:simple-simplicial}, either $C_1$ or $C_1^*$ has a non-classical facet. Then, by Lemma~\ref{lemma:facet-retract} either $C_1$ or $C_1^*$ has a non-classical retract of dimension $d_1-1$, which is naturally another proper polyhedral cone. Since retracts dualise, these two facts are actually equivalent. Hence $C_1$ has a non-classical proper polyhedral retract of dimension $d_1-1$. Continuing in this way, we can reduce the dimensions $d_1$ and $d_2$ of $C_1$ and $C_2$, until we achieve $d_1=d_2=3$. The statement then follows from Result~\ref{3-dim result}.
\end{proof}

\subsection{Proof of Result~\ref{semiquantum result}} \label{sec:proof-semiquantum}

In this section we consider pairs of cones, where one element of the pair is the cone $\PSD_n$ of $n \times n$ positive semidefinite matrices with complex entries. In other words, we look at bipartite systems $AB$, where system $A$ is described by usual quantum mechanics and system $B$ is an arbitrary GPT.

Remarkably, the problem of whether such a pair of cones is entangleable is equivalent to a recently emerged question about operator systems, formulated either in terms of \emph{operator systems} or of \emph{matrix convex sets}. Before presenting our methods, we quickly review this connection. The content of the next paragraph is not essential to the understanding of the proof of Result~\ref{semiquantum result}.

As explained in~\cite{Fritz2017,Passer2018}, an operator system in $d$ variables can be described by a sequence $\mathcal{W} = (W_n)_{n \geq 1}$ of proper cones, where $W_n$ lives in the space $\Herm_n^d$ of $d$-tuples of $n \times n$ matrices. Such a sequence is asked to satisfy compatibility conditions under the action of completely positive maps. As it turns out, given a proper cone $W \subset \R^d$, there is a minimal operator system $\mathcal{W}^{\min}$ and a maximal operator system $\mathcal{W}^{\max}$ satisfying the condition $W_1^{\min} = W_1^{\max} = W$. This means that any operator system $(W_n)$ such that $W_1=W$ must satisfy $W_n^{\min} \subseteq W_n \subseteq W_n^{\max}$. Moreover, the minimal and maximal operator systems are constructed using the minimal and maximal tensor product:
\begin{align*}
W_n^{\min} &= \PSD_n \tmin W, \\
W_n^{\max} &= \PSD_n \tmax W.
\end{align*}
A major result in~\cite{Passer2018} is the proof of the fact that the equality $\mathcal{W}^{\max}=\mathcal{W}^{\min}$ between operator systems (i.e.\ between sequences of cones) is equivalent to the starting cone $W$ being classical. In other words, any non-classical theory, when coupled with quantum mechanics $\mathrm{QM}_n$ for $n$ large enough, forms an entangleable pair. Our Result~\ref{semiquantum result} lowers the value of $n$ needed to guarantee entangleability, coming closer to the conjectured value $n=2$.

Our proof of Result~\ref{semiquantum result} relies on an extremal property of the simplex in convex geometry: \emph{the simplex is the convex shape which is most different from the round ball.} Here is a precise formulation of this property. We denote by $B_d$ the unit ball in the standard Euclidean space $\R^d$. Given a convex body $K \subset \R^d$, one defines its \emph{asphericity} $a(K)$ as the ratio between the radii of inscribed and circumscribed homothetic Euclidean balls, after preprocessing by applying a suitable affine map
\begin{equation*}
a(K) \coloneqq \inf \{ r>1 \st \textnormal{there is an affine map } \Phi : \R^d \to \R^d \textnormal{ such that } B_d \subseteq \Phi(K) \subseteq rB_d \}.
\end{equation*}
The minimal value $a(K)=1$ of the asphericity corresponds to the case when $K$ is an ellipsoid, i.e.\ an affine image of $B_d$. At the other side of the spectrum, the maximal value of asphericity is achieved for simplices. 

\begin{theorem}[Simplices maximize asphericity] \label{theorem:simplices-maximize-asphericity}
Any convex body $K \subset \R^d$ satisfies the inequality $a(K) \leq d$. Moreover, $a(K)=d$ if and only if $K$ is a simplex.
\end{theorem}

The first part of Theorem~\ref{theorem:simplices-maximize-asphericity} is well-known~\cite{John48}, while the second part was proved in~\cite{Leichtweiss} and later rediscovered in~\cite{Palmon}. 

Since the asphericity is defined by comparison with a Euclidean ball, the cones over a Euclidean ball with different radii play a central role when applying Theorem~\ref{theorem:simplices-maximize-asphericity}. We introduce them as Lorentz cones, defined for $r > 0$ as
\[ \mathsf{L}_d(r) \coloneqq \left\{ (x_1,\dots,x_{d+1}) \in \R^{d+1} \, : \, \sqrt{x_1^2 + \cdots + x_d^2} \leq r x_{d+1} \right\} .\]
Note that $\mathsf{L}_d(r)$ is the cone over the ball $r B_d$, and is thus symmetric in the sense of Section~\ref{sec:proof-symmetric-result}.

\begin{lemma} \label{lemma:tensor-lorentz}
The inclusion $\mathsf{L}_d(1) \tmax \mathsf{L}_d(1) \subseteq \mathsf{L}_d(1) \tmin \mathsf{L}_d(r)$ holds if and only if $r \geq d$.
\end{lemma}

We only give here intuition behind the critical value $r=d$ which appears in Lemma~\ref{lemma:tensor-lorentz}, and refer to the Supplementary Information for a complete proof. The minimal tensor product of Lorentz cones is intimately connected with the operator norm $\|\cdot\|_{\infty}$ on matrices, and similarly the maximal tensor product of Lorentz cones is connected with the trace norm $\|\cdot\|_1$. It is well known that the inequalities 
\begin{equation} \label{eq:holder} \|\cdot\|_{\iy} \leq \|\cdot\|_1 \leq d \|\cdot\|_{\iy} \end{equation}
hold for $d \times d$ matrices, and that the value $d$ cannot be changed into a smaller number. This can be shown by plugging in Eq.~\eqref{eq:holder} the identity matrix. This is the primary reason for the appearance of the value $d$ in Lemma~\ref{lemma:tensor-lorentz}.

Our approach to prove Result~\ref{semiquantum result} uses another ingredient, which relates the Lorentz cone with the cone of positive semidefinite matrices. As in the proof of Result~\ref{polyhedral result}, retracts are a key concept.

\begin{proposition} \label{prop:lorentz-retract-PSD}
If $d \leq 2n$, then the cone $\mathsf{L}_{d}$ is a retract of $\PSD_{2^n}$.
\end{proposition}

The construction behind Proposition~\ref{prop:lorentz-retract-PSD} is based on a well-known fact: one can find $2n$ trace zero unitary matrices of size $2^n$ which pairwise anticommute. This can be either derived from the theory of Clifford algebras, or constructed by hand as tensor products of Pauli matrices.

\begin{proof}[Proof of Result~\ref{semiquantum result}]
The fact that classicality prevents entangleability is the easy direction (Lemma~\ref{lemma:easy-direction}). Therefore, consider a pair $(C,\PSD_n)$, where $C$ is a cone in dimension $d+1$, and let us show that this pair is entangleable provided that $C$ is non-classical and $\lfloor \log_2 n \rfloor \geq d/2$. Thanks to Proposition~\ref{prop:lorentz-retract-PSD}, we know that in this case $\mathsf{L}_{d}$ is a retract of $\PSD_n$.
Using the connection between entangleability and retracts explained in Proposition~\ref{proposition:retracts-nuclear}, we obtain that the pair $(C,\mathsf{L}_{d})$ is not entangleable either. 

Now, let $K$ be a $d$-dimensional convex body which is a base of the cone $C$, and denote with $r \coloneqq a(K)$ its asphericity. If we replace $K$ by a suitable affine image (which changes neither the geometry nor the entangleability properties of the cone $C$), we may assume that $B_{d} \subseteq K \subseteq rB_{d}$, or equivalently that $\mathsf{L}_{d}(1) \subseteq C \subseteq \mathsf{L}_d(r)$. We now write
\begin{equation} \label{eq:proof-semiquantum}
\mathsf{L}_d(1)\tmax \mathsf{L}_d(1) 
\textsubseteq{(i)} \mathsf{L}_d(1)\tmax C 
\texteq{(ii)} \mathsf{L}_d(1) \tmin C
\textsubseteq{(iii)} \mathsf{L}_d(1) \tmin \mathsf{L}_d(r),
\end{equation}
where (i) and (iii) follow from the fact that $\tmin$ and $\tmax$ are increasing operations with respect to set inclusion, and (ii) expresses the unentangleability of the pair $(C,\mathsf{L}_d)$. By Lemma~\ref{lemma:tensor-lorentz}, the inclusion $\mathsf{L}_d(1)\tmax \mathsf{L}_d(1) \subseteq \mathsf{L}_d(1) \tmin \mathsf{L}_d(r)$ implies that $r \geq d$. This means that $K$ has asphericity at least $d$. By Theorem~\ref{theorem:simplices-maximize-asphericity}, this is only possible if $K$ is a $d$-dimensional simplex, and therefore the corresponding cone $C$ is classical.
\end{proof}

\subsection{Proof of Result~\ref{symmetric result}} \label{sec:proof-symmetric-result}

Consider a GPT $(V,C,u)$ whose state space $\Omega$ is symmetric with respect to a centre $\gamma \in \Omega$. We can decompose the vector space $V$ as $V=\R \oplus X$, where $X\coloneqq \ker(u)$ is the kernel of $u$. The state space defines a norm on $X$ through the choice $B_{X}\coloneqq \Omega - \gamma \subset X$ for the unit ball. Accordingly, every state can be written as $\omega = \gamma + x$, where $x \in X$ satisfies $\|x\|_{X}\leq 1$. We can define the projection $\Pi:V\to X$ onto $X$ via the formula $\Pi(v)\coloneqq v - u(v) \gamma$, so that with the above notation $\Pi(\omega)=x$. 

From the above discussion it appears that there is a natural connection between normed spaces and symmetric proper cones. Since we want to understand the properties of the latter under tensor products, we need to first review the known properties of the former. Given two finite-dimensional normed vector spaces $X,Y$, there are at least two canonical norms that these induce on the tensor product $X\otimes Y$, namely, the injective tensor norm $\|\cdot\|_{X\otimes_\varepsilon Y}$ and the projective tensor norm $\|\cdot\|_{X\otimes_\pi Y}$. For an arbitrary $z\in X\otimes Y$, these are given by~\cite{DEFANT}
\begin{align}
    \|z\|_{X\otimes_\varepsilon Y} \coloneqq&\ \sup\left\{ (f\otimes g)(z):\, f\in B_{X^*},\, g\in B_{Y^*}\right\} , \label{inj} \\
    \|z\|_{X\otimes_\pi Y} \coloneqq&\ \inf\left\{ \sumno_i \|x_i\|_X \|y_i\|_Y:\, z=\sumno_i x_i\otimes y_i\right\} . \label{proj}
\end{align}
Here, $B_{X^*}$ denotes the unit ball of the dual space $X^*$, whose corresponding norm is defined by the expression $\|f\|_{X^*}\coloneqq \sup_{x\in X\backslash \{0\}} \frac{|f(x)|}{\|x\|}$.

It is not difficult to verify directly that the inequality $\|\cdot\|_{X\otimes_\varepsilon Y}\leq \|\cdot\|_{X\otimes_\pi Y}$ holds in full generality, with equality for product tensors. Moreover, since the space $X\otimes Y$ is finite-dimensional, and all norms on a finite-dimensional space are equivalent, there will exist a constant $\rho(X,Y)$, which depends only on $X$ and $Y$, which makes the opposite inequality also true: $\|\cdot\|_{X\otimes_\pi Y}\leq \rho(X,Y) \|\cdot\|_{X\otimes_\varepsilon Y}$. The \emph{minimal} such constant across all normed spaces of fixed dimension $n,m$ is a universal function of these two integers alone, called the \emph{projective/injective ratio} and denoted by $r(n,m)$. By definition, for every pair of normed spaces $X$ and $Y$ of dimensions $\dim X=n$ and $\dim Y=m$, there exists a tensor $z\in X\otimes Y$ with $\|z\|_{X\otimes_\varepsilon Y} = 1$ such that
\begin{equation}
    \|z\|_{X\otimes_\pi Y} \geq r(n,m) \|z\|_{X\otimes_\varepsilon Y} = r(n,m)\, .
    \label{gap inj proj}
\end{equation}
The function $r(n,m)$ was defined and studied in~\cite{XOR}, whose results find here a novel application. Let us stress that it is not even clear a priori that one should have $r(n,m)>1$ for all $n,m >1$. That this indeed is the case was one of the main findings of~\cite{XOR}.

Since injective and projective tensor norms always coincide on product tensors, we may conjecture that any tensor $z$ such that $\|z\|_{X\otimes_\pi Y}>\|z\|_{X\otimes_\varepsilon Y}$ may in fact be `entangled' in some sense. To make this statement rigorous and quantitative, we need two ingredients: (i)~an entanglement measure for states of a bipartite GPT; and (ii)~a systematic way of evaluating such a measure in terms of tensor norms. To address (i) we look at the entanglement robustness, which was defined in~\cite{VidalTarrach} for quantum states, and that we can immediately extend to the GPT setting~\cite{Takagi2019}. Let $(V_1,C_1,u_1)$ and $(V_2,C_2,u_2)$ be two GPTs. For a candidate bipartite state $\omega\in C_1\tmax C_2$, the \emph{entanglement robustness} is defined as the minimal amount of separable noise that makes a state separable, in formula
\begin{equation} \label{ent-rob}
\erob (\omega)\coloneqq \min\left\{ (u_1\otimes u_2)(\zeta):\, \zeta,\, \omega+\zeta\in C_1\tmin C_2 \right\} .
\end{equation}
We believe that this entanglement measure, whose definition is rooted in convex geometry alone, is the natural choice in the context of GPTs. To complete our programme we need to tackle problem (ii) above. This is done by means of the following lemma.

\begin{lemma} \label{lemma:ent-rob}
Let $(V_1,C_1,u_1),\, (V_2,C_2,u_2)$ be two symmetric GPTs. Call $\gamma_1,\gamma_2$ the centres of the state spaces, and $X_1,X_2$ the associated normed spaces. For $z \in X_1\otimes X_2$, consider the normalised state $\omega(z)\coloneqq \gamma_1\otimes \gamma_2 + z$. Whenever $z$ satisfies $\|z\|_{X_1\otimes_\varepsilon X_2}\leq 1$, it holds that $\omega(z)\in C_1\tmax C_2$. In this case,
\begin{equation} \label{ent rob estimate}
\erob \left(\omega(z)\right) \geq \frac{\|z\|_{X_1\otimes_\pi X_2}-1}{2}\, .
\end{equation}
\end{lemma}

We are finally ready to prove Result~\ref{symmetric result}.

\begin{proof}[Proof of Result~\ref{symmetric result}]
Consider a pair of symmetric GPTs of dimensions $n+1,\,m+1\geq 3$. Combining Eq.~\eqref{ent rob estimate} and Eq.~\eqref{gap inj proj}, we see that there is a normalised state $\omega$ in the maximal tensor product $C_1\tmax C_2$ such that
\begin{equation}
    \erob(\omega)\geq \frac{r(n,m)-1}{2}\, .
\end{equation}
The claim follows from the estimates $r(n,m)\geq 19/18$~\cite[Theorem~2]{XOR}, valid for all $n,m\geq 2$, and $r(n,m)\geq c \min\{n,m\}^{1/8-o(1)}$~\cite[Theorem~6]{XOR}, valid in the limit $n,m\to\infty$.
\end{proof}

\section{Conclusions}

In this work, we defined and studied the model-independent notion of universal entanglement in the context of general probabilistic theories. We have shown that the failure of the local state spaces to have the geometric shapes of simplices, which is a manifestation of the existence of superpositions, is intimately connected with the existence of entanglement at the level of bipartite states or measurements. This connection, which before was thought of as an accident of the quantum formalism, is elevated here to a foundational status. In fact, our main conjecture states that all pairs of non-classical GPTs can be entangled by composition.

A mathematically equivalent version of this problem is already implicit in the work of Namioka and Phelps~\cite{NP}, and was systematically studied in the 1970s by Barker~\cite{Barker1976,Barker-review}. It consists in proving that all pairs of non-classical cones are such that the maximal tensor product is strictly larger than the minimal. The motivation driving all these previous efforts was of a fundamentally order-theoretical nature, and not related to entangleability of GPTs.

We presented strong evidence in favour of our main conjecture, proving it in a number of physically relevant cases. Namely, we showed that it is true when both cones are $3$-dimensional (Result~\ref{3-dim result}) or polyhedral (Result~\ref{polyhedral result}), and also when one of the local theories is quantum mechanics on a Hilbert space of a sufficiently large dimension (Result~\ref{semiquantum result}). We also took one step further and put forth a quantitative extension of the above qualitative conjecture. Namely, we proposed that the maximal entanglement exhibited by a pair of theories may be lower bounded by a universal function of their departure from classicality, as measured by an appropriate quantifier. We presented evidence in support of this hypothesis, proving it for the geometrically vast class of symmetric cones (Result~\ref{symmetric result}). We briefly discussed further extensions of this approach to non-locality in bipartite GPT systems.

On the mathematical side, our results mark the first progress on the conjecture since the work of Barker, more than 40 years ago. Moreover, our methods are substantially innovative: we put to good use the order-theoretic concept of retract, devised general techniques to construct entangled states in bipartite GPTs, and further investigated connections between functional analysis and general probabilistic theories, as already developed in~\cite{ultimate, XOR}.

In conclusion, our work sheds new light on a seemingly accidental connection between the notions of local superposition and global entanglement, promoting it to a logically unavoidable implication. This prompts us to reconsider the status of entanglement as a fundamental ingredient of Nature.



\section*{Acknowledgements}

We thank Andreas Winter and Stanis\l aw Szarek for many enlightening discussions, and Martin Pl\'{a}vala for suggesting using the concept of a retract. This work was partly achieved during our visit in Institut Henri Poincar\'e, which we thank for hospitality. GA was supported in part by ANR (France) under the grant StoQ (2014-CE25-0003). LL acknowledges financial support from the European Research Council (ERC) under the Starting Grant GQCOP (Grant no.~637352). CP is partially supported by the Spanish `Ram\'on y Cajal Programme' (RYC-2012-10449), the Spanish `Severo Ochoa Programme' for Centres of Excellence (SEV-2015-0554) and the grant MTM2017-88385-P, funded by Spanish MEC.

\bibliography{cones}

\clearpage

\onecolumngrid
\begin{center}
\vspace*{\baselineskip}
{\textbf{\large Supplemental Material}}\\
\end{center}


\renewcommand{\theequation}{S\arabic{equation}}
\renewcommand{\thetheorem}{S\arabic{theorem}}
\renewcommand{\thedefinition}{S\arabic{definition}}
\setcounter{equation}{0}
\setcounter{theorem}{0}
\setcounter{figure}{0}
\setcounter{table}{0}
\setcounter{section}{0}
\setcounter{page}{1}
\makeatletter

We present here the formal proofs of our results. The focus is on technical precision; we refer to the main article for motivation. Statements which appear only in Supplementary Information are labelled by S1, S2, etc.; statements which are duplicated from the main article use the same label as in the main article. In accordance with standard usage in mathematical literature, we rephrased our main results as theorems, keeping the same labels.

Section~\ref{section:intro} introduces all concepts which are needed to define the tensor products of cones, and restatements of the results announced in the main article. Section~\ref{proof:dim3} contains the proof of Result~\ref{3-dim result} on $3$-dimensional cones. Section~\ref{proof:polyhedral} is devoted to the proof of Result~\ref{polyhedral result}, concerning polyhedral cones. Section~\ref{section:semiquantum} contains the proof of Result~\ref{semiquantum result}, which deals with the case of one cone being that of positive semidefinite matrices. In Section~\ref{section:symmetric} we prove Result~\ref{symmetric result} on symmetric cones. Finally, Section~\ref{section:retracts} contains extra information about retracts of cones, a concept which plays a central role in our argument.

\section{Definitions, elementary facts and statements of the theorems}
\label{section:intro}

\subsection{Convexity, convex cones}

All vector spaces are assumed to be over the real field, and finite-dimensional. Hereafter, we denote with $\R_+$ the set of non-negative real numbers.

\begin{definition}
A \emph{cone} is a subset $\C$ of a vector space with the following property: for every
$x \in \C$ and $\alpha \in \R_+$, we have $\alpha x \in \C$.
\end{definition}

\begin{definition}
Let $V$ be a vector space. 
    \begin{enumerate}[label={(S\arabic{definition}.\arabic*)}, leftmargin=\widthof{(S2.7)}+\labelsep]
        \item The \emph{dual space} to $V$, denoted $V^*$, is defined as the space of linear maps
        from $V$ to $\R$. We always identify the double dual $(V^*)^*$ with $V$.
        \item A subset $A \subseteq V$ is \emph{convex} if $x$, $y \in A$ implies $\lambda x + (1-\lambda) y \in A$ for every $\lambda \in [0,1]$.
        It follows that a subset $\C \subset V$ is a convex cone if and
        only if $x$, $y \in \C$ implies $\alpha x + \beta y \in \C$ for every $\alpha$, $\beta \in \R_+$.
        \item \label{def:convex-body} A \emph{convex body} in $V$ is a compact convex set with nonempty interior.
        \item The \emph{convex hull} of a subset $A \subseteq V$ is 
        \[ \conv(A)  \coloneqq \left\{ \sum_{i=1}^n \lambda_i a_i \st n \in \{1,2,3,\ldots\}, \ \lambda_i \in [0,1], \ a_i \in A, \ \sum_{i=1}^n \lambda_i=1 \right\}. \]         
        Equivalently, it equals the intersection of all convex sets containing $A$. 
        
        \item The \emph{affine span} of a subset $A \subseteq V$ is
        \[ \aff(A)  \coloneqq \left\{ \sum_{i=1}^n \lambda_i a_i \st n \in \{1,2,3,\ldots\}, \ \lambda_i \in \R, \ a_i \in A, \ \sum_{i=1}^n \lambda_i=1  \right\}. \]
        Equivalently, it equals the intersection of all affine subspaces containing $A$.
        \item A set $\{y_1,\dots,y_n\} \subset V$ is \emph{affinely independent} if $y_i \not \in 
        \aff\{y_j :\, j \neq i\}$ for every $1 \leq i \leq n$. 
        
        \item The \emph{conical hull} of a subset $A \subseteq V$ is
        \[ \cone(A)  \coloneqq \left\{ \sum_{i=1}^n \lambda_i a_i \st n \in \{1,2,3,\cdots\}, \ \lambda_i \in \R_+, \ a_i \in A  \right\}. \]
        Equivalently, it equals the intersection of all convex cones containing $A$.
\end{enumerate}
\end{definition}

\begin{definition}
Let $K$ be a convex set in a vector space.
\begin{enumerate}[label={(S\arabic{definition}.\arabic*)}, leftmargin=\widthof{(S3.5)}+\labelsep]
\item The \emph{dimension} of $K$, denoted $\dim (K)$, is defined as the dimension of its affine span.
\item A nonempty convex subset $F \subseteq K$ is a \emph{face} of $K$ if $x \in K$, $y \in K$, $0 < \lambda <1$ and $\lambda x +(1-\lambda)y \in F$ imply $x$, $y \in F$. 
\item A face $F \subseteq K$ is \emph{proper} if $F \neq K$. 
\item A face $F \subseteq K$ is a \emph{facet} if $\dim (F) = \dim (K) - 1$.
\item An element $x \in K$ is an \emph{extreme point} of $K$ if $\{x\}$ is a face of $K$. This is equivalent to say that the equation $x=\lambda y + (1-\lambda)z$, for $y,z\in K$ and $0 < \lambda < 1$, implies that $y=z=x$.
\end{enumerate}
\end{definition}

\begin{definition}
Let $\C$ be a cone. An \emph{extreme ray} of $\C$ is a face of dimension $1$. Equivalently, for $x \in \C \setminus \{0\}$, the set $\R_+ x$ is an extreme ray of $\C$ if the equation $x=y+z$ for $y$, $z \in \C$
implies $y=\alpha x$ for some $\alpha \in [0,1]$.
\end{definition}

\begin{definition}
Let $V_1$, $V_2$ be vector spaces, and $\C_1 \subset V_1$, $\C_2 \subset V_2$ be convex cones. The cones $\C_1$ and $\C_2$ are \emph{isomorphic} if there is a linear bijection $\Phi : V_1 \to V_2$ such that $\Phi(\C_1) = \C_2$.
\end{definition}

\begin{definition}
 Let $\C$ be a cone in a vector space $V$.
    \begin{enumerate}[label={(S\arabic{definition}.\arabic*)}, leftmargin=\widthof{(S6.4)}+\labelsep]    
        \item $\C$ is \emph{salient} if $\C \cap (-\C) = \{0\}$.
        \item $\C$ is \emph{generating} if the linear span of $\C$ equals $V$.
        \item $\C$ is \emph{proper} if it is convex, closed, salient and generating.
        \item The \emph{dual cone} $\C^* \subset V^*$ is defined as $\C^* \coloneqq  \{ f \in V^* :\, f(x) \geq 0 \textnormal{ for every } x \in \C \}$. 
    \end{enumerate}
\end{definition}

\begin{fact}[Bipolar theorem~{\cite[Theorem 14.1]{Rockafellar70}}] \label{fact:bidual}
Every closed convex cone $C\subseteq V$ satisfies $C^{**}=C$ upon the identification $V^{**}=V$.
\end{fact}

\begin{definition}[GPT] \label{def:GPT}
A \emph{general probabilistic theory} (GPT) is a triple $(V,C,u)$, where $V$ is a finite-dimensional real vector space, $C\subset V$ is a proper cone, and $u\in \inter (C^*)$ is a positive functional on $C$. The corresponding \emph{state space} is $\Omega\coloneqq C \cap u^{-1}(1)$. We call $\dim (V)$ the \emph{dimension} of the GPT.
\end{definition}

\begin{fact}[{see~\cite[Theorem 3.5]{ALIPRANTIS}}] \label{fact:u-strictly-positive}
Given a proper cone $C$, the interior $\inter(C^*)$ of the dual cone $C^*$ coincides with the set of strictly positive functionals on $C$. A functional $\varphi\in V^*$ is called \emph{strictly positive} on a cone $C\subseteq V$ if $\varphi(x)\geq 0$ for all $x\in C$, with equality only for $x=0$.
\end{fact}

Hence, in Definition~\ref{def:GPT} we could equivalently require that $u$ be a strictly positive functional. 

\begin{definition} \label{def:base}
A \emph{base} of a cone $C\subseteq V$ is a convex subset $K\subset C$ such that for all $x\in C$ there is a unique $t\geq 0$ that satisfies $x\in tK$.
\end{definition}

\begin{fact}[{see~\cite[Theorem 1.47]{ALIPRANTIS}}] \label{fact:bases-strictly-positive}
Let $C$ be a convex and salient cone. Then all bases of $C$ (if they exist) are of the form $K = C\cap \varphi^{-1}(1)$, for some strictly positive functional $\varphi$ on $C$. 
\end{fact}

In particular, note that the state space $\Omega$ of a GPT $(V,C,u)$ is always a base of the cone $C$. There is a natural yet general way to construct a cone with a given base:

\begin{definition}
Let $V$ be a vector space, and $K \subseteq V$ a convex set. The \emph{cone with base} $K$ is the cone in $V \oplus \R$ defined as
\begin{equation} \label{eq:cone-over-K} \CC(K) = \{ (tx,t) \st x \in K, \ t \in \R_+ \}. \end{equation}
\end{definition}

Fact~\ref{fact:cone-has-base} shows that any proper cone has a base.

\begin{fact}[see {\cite[Corollary 1.8]{ABMB}}] \label{fact:cone-has-base}
Let $C\subseteq V$ be a finite-dimensional convex, salient and generating salient cone. Then:
\begin{enumerate}[(a)]
\item $C$ is closed (and hence, proper) if and only if bases for it exist and are all compact;
\item if $C$ is proper, then it is isomorphic to $\CC(K)$, where $K$ is any of its bases;
\item in particular, if $C$ is proper and $\varphi\in \inter(C^*)$ is strictly positive on $C$, then $C\cap \varphi^{-1}(1)$ is a convex body inside the affine space $\varphi^{-1}(1)$.
\end{enumerate}
\end{fact}

In particular, note that whenever $(V,C,u)$ is a GPT, the cone $C$ is isomorphic to $\CC(\Omega)$, where $\Omega$ is the state space. Also, $\Omega$ is a convex body when viewed as a subset of the affine space $u^{-1}(1)$.
Fact~\ref{fact:faces-cone-base} relates the facial structure of a cone and of its base.

\begin{fact}[see {\cite[Proposition 1.9]{ABMB}}] \label{fact:faces-cone-base}
Let $K$ be a convex body. There is a one-to-one correspondence between faces of $K$ and faces of $\CC(K)$ distinct from $\{0\}$, given by the map $F \mapsto \R_+ F$.
\end{fact}

\begin{fact} \label{fact:affine-linear}
Let $K$ be a convex body in a vector space $V$, and let $\Phi : V \to V$ be an affine bijection, i.e.\ a map of the form $x \mapsto \Phi(x) = \Psi(x) + z$, where $z \in V$ and $\Psi : V \to V$ is an invertible linear map. Then the cones $\CC(K)$ and $\CC(\Phi(K))$ are isomorphic.
\end{fact}

\begin{proof}
One checks that the linear map $\tilde{\Phi} : V \times \R \to V \times \R$ defined by $\tilde{\Phi}(x,t) = (\Psi(x)+t z,t)$ is invertible and satisfies
$\tilde{\Phi}(\CC(K)) = \CC(\Phi(K))$.
\end{proof}

\begin{definition}
A cone is \emph{classical} if it is isomorphic to 
\begin{equation}
\R_+^n \coloneqq \{ (x_1,\dots,x_n) \in \R^n \st x_i \geq 0 \textnormal{ for } 1 \leq i \leq n\}
\end{equation}
for some integer $n \geq 1$.
\end{definition}

\begin{fact} \label{fact:dual-classical}
Let $\C$ be a proper cone. Then $\C$ is classical if and only if $\C^*$ is classical.
\end{fact}

\begin{definition}
A \emph{simplex} in a vector space is the convex hull of an affinely independent set. 
\end{definition}

\begin{fact} \label{fact:classical}
Let $\C$ be a proper cone in a vector space $V$. The following are equivalent:
\begin{enumerate}
    \item[(i)] $\C$ is classical,
    \item[(ii)] there exists a basis $A$ of the vector space $V$ such that $\C = \cone(A)$,
    \item[(iii)] there exists a simplex $\Delta$ with $\dim(\Delta)=\dim(V)$ such that $\C$ is isomorphic to $\CC(\Delta)$,
    \item[(iv)] every convex set $K$ such that $\CC(K)$ is isomorphic to $\C$ is a simplex.
\end{enumerate}
\end{fact}

\begin{definition}
Let $V$, $W$ be vector spaces. The \emph{adjoint} of a linear map $\Phi : V \to W$ is the linear map
$\Phi^* : W^* \to V^*$ defined by the relation $\Phi^*(g)(x)=g(\Phi(x))$ for every $x \in V$, $g \in W^*$.
\end{definition}

\begin{fact} \label{fact:polar-linearmap}
Let $V$, $W$ be vector spaces, $\C \subseteq V$ be a convex cone and $\Phi : V \to W$ be a linear map. Then $\Phi(\C)^* = (\Phi^*)^{-1}(\C^*)$.
\end{fact}

\begin{proof}
For $g \in W^*$, we have the equivalences
\[ g \in \Phi(\C)^* \iff \forall x \in \C,\ g(\Phi(x)) \geq 0 \iff
\forall x \in \C,\ \Phi^*(g)(x) \geq 0 \iff \Phi^*(g) \in \C^*, 
\]
hence the result follows.
\end{proof}

\subsection{Tensor products of cones}

We now introduce our main object of study: tensor products of cones.

\begin{definition}
Let $V_1$ and $V_2$ be vector spaces, and $\C_1 \subseteq V_1$, $\C_2 \subseteq V_2$ be convex cones.
\begin{enumerate}[label={(S\arabic{definition}.\arabic*)}, leftmargin=\widthof{(S17.2)}+\labelsep]
\item The \emph{minimal tensor product} of $\C_1$ and $\C_2$ is the convex cone in $V_1 \otimes V_2$ defined by
\begin{equation*}
\C_1 \tmin \C_2 \coloneqq \conv \{ x_1 \otimes x_2 \st x_1 \in \C_1,\  x_2 \in \C_2 \}\, .
\end{equation*}
\item The \emph{maximal tensor product} of $\C_1$ and $\C_2$ is the convex cone in $V_1 \otimes V_2$ defined by
\begin{equation*}
\C_1 \tmax \C_2 \coloneqq \{ x \in V_1 \otimes V_2 \st (f_1 \otimes f_2)(x) \geq 0 \textnormal{ for every } f_1 \in \C_1^*,\  f_2 \in \C_2^* \}\, .
\end{equation*}
In this definition we identify $(V_1 \otimes V_2)^*$ and $V_1^* \otimes V_2^*$.
\end{enumerate}
\end{definition}

\begin{fact} \label{fact:tmin-is-proper}
Let $\C_1$ and $\C_2$ be convex cones. Then
\begin{enumerate}[(a)]
    \item $\C_1 \tmin \C_2 \subseteq \C_1 \tmax \C_2$;
    \item if $\C_1$ and $\C_2$ are proper, then so are $\C_1 \tmin \C_2$ and $\C_1 \tmax \C_2$.
\end{enumerate}
\end{fact}

A non-obvious point in Fact~\ref{fact:tmin-is-proper} is that $\C_1 \tmin \C_2$ is closed. This can be seen by taking compact bases $K_1$, $K_2$, and by checking that $\C_1 \tmin \C_2$ also has a compact base, namely the image of the compact set $K_1 \times K_2$ under the continuous map $(x,y) \mapsto x \otimes y$.

\begin{fact} \label{fact:duality}
Minimal and maximal tensor product are dual to each other. Namely, for all proper cones $\C_1$ and $\C_2$, it holds that
\begin{align}
\left(C_1\tmin C_2\right)^* &= C_1^* \tmax C_2^*\, , \label{dual min is max} \\
\left(C_1\tmax C_2\right)^* &= C_1^* \tmin C_2^*\, . \label{dual max is min}
\end{align}
\end{fact}

\begin{definition}
Let $\C_1$ and $\C_2$ be proper cones. The pair $(\C_1,\C_2)$ is called \emph{nuclear} if $\C_1 \tmin \C_2 = \C_1 \tmax \C_2$ and \emph{entangleable} if $\C_1 \tmin \C_2 \subsetneq \C_1 \tmax \C_2$.
\end{definition}

The terminology ``nuclear'' is borrowed from the language of $C^*$-algebras, where the analogous concept has played a central role in the theory (see e.g.~\cite[\S 12]{Pisier12}).

\begin{fact} \label{fact:nuclear-bijection}
If $\C_1$ is isomorphic to $\C'_1$ and $\C_2$ is isomorphic to $\C'_2$, then
\[ (\C_1,\C_2) \textnormal{ is nuclear} \iff
 (\C'_1,\C'_2) \textnormal{ is nuclear.} \]
\end{fact}

As explained in the main text, it is natural to conjecture that a pair $(\C_1,\C_2)$ of proper cones is nuclear if and only if $\C_1$ or $\C_2$ is classical. This question can be traced back to~\cite{Barker1976, Barker-review}. The `if' direction is easy and we have already seen a proof of it in the Methods section. The original argument seem to go back to Namioka and Phelps~\cite{NP}, and to have been re-discovered many times, e.g.\ by Barker~\cite{Barker1976, Barker-review}.
%
As for the much more challenging `only if' direction, our first result settles the situation at least for $3$-dimensional cones.

\begin{manualthm}{\ref{3-dim result}}[Proved in Section~\ref{proof:dim3}] 
Let $\C_1$ and $\C_2$ be proper cones of dimension $3$. If $(\C_1,\C_2)$ is nuclear, then either $\C_1$ or $\C_2$ is classical.
\end{manualthm}

Our second result solves the problem for polyhedral cones.

\begin{definition}
A convex cone is \emph{polyhedral} if it is the conical hull of a finite set.
\end{definition}

\begin{manualthm}{\ref{polyhedral result}}[Proved in Section~\ref{proof:polyhedral}] 
Let $\C_1$ and $\C_2$ be proper polyhedral cones. If $(\C_1,\C_2)$ is nuclear, then either $\C_1$ or $\C_2$ is classical.
\end{manualthm}

Our third result concerns the case when one cone is the cone of PSD matrices.

\begin{definition} \label{def:PSD}
We denote by $\PSD_n$ be the cone of $n\times n$ positive semidefinite matrices, which is contained in the real vector space $\gM_n^{\mathrm{sa}}$ of $n \times n$ Hermitian matrices with complex entries. Note that $\dim (\PSD_n)=n^2$.
\end{definition}

\begin{manualthm}{\ref{semiquantum result}} 
Let $\C$ be a proper cone of dimension $d$, and assume that $\floor{\log_2 n} \geq \frac{d-1}{2}$. Then the pair $(\C, \PSD_n)$ is nuclear if and only if $\C$ is classical.
\end{manualthm}

Our methods also yield a proof of the following result from~\cite{Huber-Netzer}.
\begin{manualthm}{\ref{semiquantum result}'} 
Let $\C$ be a proper polyhedral cone, and assume that $n \geq 2$. Then the pair $(\C, \PSD_n)$ is nuclear if and only if $\C$ is classical.
\end{manualthm}


Our last result is a first attempt towards a quantitative extension of the above qualitative correspondence between local non-classicality and global entangleability. We are able to rigorously establish such an extension in the geometrically natural case where the local state spaces are centrally symmetric around some fixed local states.

\begin{definition}
A GPT $(V,C,u)$ is called \emph{symmetric} if there is a state $\gamma\in \Omega\coloneqq C\cap u^{-1}(1)$, called its \emph{centre}, such that $2\gamma - \omega\in \Omega$ for all $\omega\in \Omega$.
\end{definition}

We also put forth the following general definition of entanglement robustness.

\begin{definition}
Let $(V_1,C_1,u_1),\, (V_2,C_2,u_2)$ be GPTs. The \emph{entanglement robustness} of $\omega\in C_1\tmax C_2$ is given by
\begin{equation*}
\erob (\omega)\coloneqq \min\left\{ (u_1\otimes u_2)(\zeta):\, \zeta,\, \omega+\zeta\in C_1\tmin C_2 \right\} . \tag{\ref{ent-rob}}
\end{equation*}
\end{definition}


\begin{manualthm}{\ref{symmetric result}}
Let $(V_1,C_1,u_1)$ and $(V_2,C_2,u_2)$ be two symmetric GPTs of dimensions $n+1$ and $m+1$, respectively. Then there exists $\omega \in C_1 \tmax C_2$ such that $\erob(\omega) \geq c \min\{n,m\}^{1/8-o(1)}$, where $c>0$ is a number, and $o(1)$ denotes a quantity tending to $0$ as $\min \{n,m\}$ tends to infinity. Moreover, provided $n,m \geq 2$, there exists $\omega \in C_1 \tmax C_2$ such that $\erob(\omega) \geq 1/36$.
\end{manualthm}

\subsection{Retracts}

\begin{definition} Let $V$, $W$ be vector spaces. 
\begin{enumerate}[label={(S\arabic{definition}.\arabic*)}, leftmargin=\widthof{(S23.2)}+\labelsep]
\item We denote by $\gL(V,W)$ the vector space of linear maps from $V$ to $W$.
\item We denote by $z \mapsto \mathrm{op}(z)$ the canonical bijection between $V \otimes W$ and
$\gL(V^*,W)$, defined for $x \in V$, $y \in W$, $f \in V^*$ by $\mathrm{op}(x \otimes y) : f \mapsto
f(x)y$. 
\end{enumerate}
\end{definition}

\begin{definition}
Let $\C_1 \subseteq V_1$, $\C_2 \subseteq V_2$ be cones in vector spaces. A linear map $\Phi \in \gL(V_1,V_2)$ is  $(\C_1,\C_2)$-\emph{positive} (or \emph{positive} if there is no ambiguity) if $\Phi(\C_1) \subseteq \C_2$. We denote by $\gPos(\C_1,\C_2) \subseteq \gL(V_1,V_2)$ the cone of $(\C_1,\C_2)$-positive maps.
\end{definition}

\begin{fact} \label{fact:adjoint-dual}
Let $\C_1$ and $\C_2$ be proper cones. If a linear map $\Phi$ is $(\C_1,\C_2)$-positive, then the adjoint map $\Phi^*$ is $(\C_2^*,\C_1^*)$-positive.
\end{fact}

Fact~\ref{fact:tmax-and-pos} interprets the maximal tensor product as a cone of positive maps.

\begin{fact} \label{fact:tmax-and-pos}
Let $\C_1$ and $\C_2$ be proper cones. Then
\[ \mathrm{op}(\C_1 \tmax \C_2) = \gPos(\C_1^*,\C_2).\]
\end{fact}

We now introduce the concept of a \emph{retract}.

\begin{definition}
Let $\C \subseteq V$ and $\C' \subseteq V'$ be convex cones in vector spaces. We say that $\C'$ is a \emph{retract} of $\C$ if there are positive maps $\Phi \in \gPos(\C,\C')$ and $\Psi \in \gPos(\C',\C)$ such that $\Phi \circ \Psi = \Id_{V'}$. This implies in particular that $\C'=\Phi(\C)$. In this context, the map $\Phi$ is called a \emph{retraction}. 
\end{definition}

\begin{fact}
If $C$ is proper and $C'$ is a retract of $C$, then also $C'$ is proper.
\end{fact}

\begin{proof}
Since $C'$ is assumed to be a convex cone, we just need to check that is closed, salient, and generating. Closedness follows from the fact that $C'=\Phi(C)$, with $C$ closed and $\Phi$ linear. To prove that $C'$ is salient, note that 
\begin{equation*}
C'\cap (-C') = (\Phi\circ\Psi)\left( C'\cap (-C')\right) \subseteq \Phi\left( \Psi(C')\cap \Psi(-C') \right) \subseteq \Phi\left( C\cap (-C)\right) = \{0\}\, .
\end{equation*}
To show that it is generating instead, write
\begin{equation*}
V'= \Phi(V) = \Phi(C-C)\subseteq \Phi(C) - \Phi(C) \subseteq C'-C'\, ,
\end{equation*}
which naturally implies that $V'= C'-C'$.
\end{proof}

Fact~\ref{fact:retract-dual} shows that retractions nicely dualise.

\begin{fact} \label{fact:retract-dual}
Let $\C_1$, $\C_2$ be proper cones. If $\C_1$ is a retract of $\C_2$, then $\C_1^*$ is a retract of $\C_2^*$. 
\end{fact}

\begin{proof}
There are maps $\Phi \in \gPos(\C_2,\C_1)$ and $\Psi \in \gPos(\C_1,\C_2)$ such that $\Phi \circ \Psi$ is the identity. As a consequence of Fact~\ref{fact:adjoint-dual}, we have that $\Phi^* \in \gPos(\C_1^*,\C_2^*)$ and $\Psi^* \in \gPos(\C_2^*,\C_1^*)$. Since $\Psi^* \circ \Phi^*= (\Phi\circ \Psi)^*$ is the identity, this shows that $\C_1^*$ is a retract of $\C_2^*$.
\end{proof}

We also check that a retract of a cone can always be realised as a section.

\begin{fact} \label{fact:retract-as-section}
Let $\C_1 \subset V_1$, $\C_2 \subset V_2$ be proper cones in vector spaces. Then $\C_1$ is a retract of $\C_2$ if and only if $\C_1$ is isomorphic to $\C_2 \cap E$, where $E \subset V_2$ is a subspace for which there is a projection $P:V_2 \to E$ such that $\C_2 \cap E = P(\C_2)$.
\end{fact}

\begin{proof}
The `if' direction is easy. Conversely, suppose that $\C_1$ is a retract of $\C_2$; consider $\Phi \in \gPos(\C_2,\C_1)$ and $\Psi \in \gPos(\C_1,\C_2)$ such that $\Phi \circ \Psi = \Id_{V_1}$. Set $E$ to be the range of $\Psi$. Since $\Psi$ is injective, the cones $\C_1$ and $\Psi(\C_1)$ are isomorphic. Finally, one checks that $P \coloneqq \Psi \circ \Phi$ is a projection onto $E$ such that $P(\C_2)=\C_2 \cap E=\Psi(\C_1)$.
\end{proof}

The concept of a retract plays an important role in our study, because of the following property, whose proof was already provided in the Methods section of the main article.

\begin{manualprop}{\ref{proposition:retracts-nuclear}}[Nuclearity passes to retracts]
Let $(\C_1,\C_2)$ be a nuclear pair. If $\C'_1$ is a retract of $\C_1$, and $\C'_2$ is a retract
of $\C_2$, then $(\C'_1,\C'_2)$ is a nuclear pair.
\end{manualprop}

We also present two extra statements about retracts which are not needed for the proofs of the main results, but which we include as we believe they could help the reader forge their intuition. They show that while $2$-dimensional sections are always retracts, this typically never happens for higher-dimensional sections. Since we could not locate these statements elsewhere in the literature, proofs are provided in Section~\ref{section:retracts}.

\begin{proposition} \label{proposition:2-dimensional-retract}
Let $\C \subset V$ be a proper cone, and $E \subseteq V$ be a $2$-dimensional subspace which intersects the interior of $\C$. Then $\C \cap E$ is a retract of $\C$.
\end{proposition}

\begin{proposition} \label{proposition:baire-no-retract}
There is a $4$-dimensional convex cone which admits no $3$-dimensional retract.
\end{proposition}

\subsection{Tensor norms}

The proof of Result~\ref{symmetric result} requires us to familiarise with the concept of tensor norms. Here we introduce the main definitions and discuss some of their elementary implications. The interested reader is referred to the monograph~\cite{DEFANT}. 

\begin{definition}
Let $X$ be a real vector space. A \emph{norm} on $X$ is a function $\|\cdot\|:X\to \R_+$ that is: (i)~faithful, meaning that $\|x\|=0$
if and only if $x=0$; (ii)~absolutely homogeneous, i.e.\ such that $\|\lambda x\|=|\lambda| \|x\|$ for all $\lambda\in \R$; and (iii)~obeys the triangle inequality, which states that $\|x+y\|\leq \|x\|+\|y\|$, for all $x,y\in X$. A vector space equipped with a norm is called a \emph{normed space}.
\end{definition}

\begin{note}
We will often specify as a subscript the normed space a norm refers to. Accordingly, the norm of $x\in X$ will be denoted by the symbol $\|x\|_X$.
\end{note}

\begin{definition}
The unit ball of a normed space $X$ is the convex set $B_X\coloneqq \left\{x\in X:\, \|x\|_X\leq 1\right\}$.
\end{definition}

\begin{definition} \label{def:dual-norm}
Let $X$ be a normed space. The dual vector space $X^*$ can be turned into a normed space itself via the definition of the \emph{dual norm}
\begin{equation} \label{dual-norm}
\|f\|_{X^*} \coloneqq \sup_{x\in B_X} |f(x)|\, .
\end{equation}
\end{definition}

\begin{fact} \label{fact:norm-dual-formula}
The bi-dual of a finite-dimensional normed space $X$ is the space $X$ itself. Namely, for all $x\in X$ it holds that
\begin{equation}
    \|x\|_X = \sup_{f\in B_{X^*}} |f(x)|\, .
\end{equation}
\end{fact}

\begin{definition}
Let $X,Y$ be finite-dimensional real vector spaces, and let $X\otimes Y$ be their tensor product.
\begin{enumerate}[label={(S\arabic{definition}.\arabic*)}, leftmargin=\widthof{(S31.2)}+\labelsep]
\item The \emph{injective tensor norm} is the norm on $X\otimes Y$ defined by
\begin{equation}
    \|z\|_{X\otimes_\varepsilon Y} \coloneqq \sup\left\{ (f\otimes g)(z):\, f\in B_{X^*},\, g\in B_{Y^*}\right\} , \tag{\ref{inj}}
\end{equation}
for all $z\in X\otimes Y$.
\item The \emph{projective tensor norm} is the norm on $X\otimes Y$ defined by
\begin{equation}
    \|z\|_{X\otimes_\pi Y} \coloneqq \inf\left\{ \sumno_i \|x_i\|_X \|y_i\|_Y:\, z=\sumno_i x_i\otimes y_i\right\} , \label{proj}
\end{equation}
for all $z\in X\otimes Y$.
\end{enumerate}
\end{definition}

\begin{fact}
Injective and projective tensor norm are dual to each other. Namely, one has the normed space identities
\begin{align}
\left( X \otimes_\varepsilon Y \right)^* &= X^* \otimes_\pi Y^*\, , \label{dual inj is proj}\\
\left( X \otimes_\pi Y \right)^* &= X^* \otimes_\varepsilon Y^*\, . \label{dual proj is inj}
\end{align}
\end{fact}

The following is easy to verify.

\begin{fact} \label{fact:inj<proj}
Let $X,Y$ be finite-dimensional spaces. Then there exists a smallest constant $\rho(X,Y)>0$ such that
\begin{equation}
    \|z\|_{X\otimes_\varepsilon Y}\leq \|z\|_{X\otimes_\pi Y} \leq \rho(X,Y) \|z\|_{X\otimes_\varepsilon Y}
\end{equation}
for all $z\in X\otimes Y$.
\end{fact}

\begin{definition}[Projective/injective ratio~\cite{XOR}] \label{def pi/epsilon ratio}
Given integers $n$, $m\geq 2$, the associated \emph{projective/injective ratio} is defined by
\begin{equation}
r(n,m) \coloneqq \inf_{\substack{\dim X =n \\[.2ex] \dim Y =m }} \rho(X,Y)\, ,
\label{r}
\end{equation}
where the infimum is over all pairs of normed spaces $X,Y$ of dimensions $n,m$, respectively, and $\rho(X,Y)$ is defined in Fact~\ref{fact:inj<proj}.
\end{definition}

We rely on the following estimates from~\cite{XOR} on the projective/injective ratio.

\begin{fact}[Theorem 6 in~\cite{XOR}] \label{fact:19/18}
For every integers $n$, $m \geq 2$, it holds that $r(n,m) \geq 19/18$.
\end{fact} 

\begin{fact}[Theorem 2 in~\cite{XOR}] \label{fact:r(n,m)}
There is a constant $c>0$ such that, for every integers $n,m \geq 2$, we have that
\begin{equation*}
r(n,m) \geq c\, \frac{\min \{n,m\}^{1/8}}{\log \min \{n,m\}}\, .
\end{equation*}
\end{fact}

It is conjectured in~\cite{XOR} that the value $19/18$ in Fact~\ref{fact:19/18} can be replaced by $\sqrt{2}$, and that the exponent $1/8$ in Fact~\ref{fact:r(n,m)} can be replaced by $1/2$.

\section{Proof of Theorem~\ref{3-dim result}: 3-dimensional cones} \label{proof:dim3}

\begin{definition}
Let $a$ and $b$ be elements of $(-1,1)$. The \emph{kite} with center $(a,b)$ is defined as 
\begin{equation*}
    T_{a,b} \coloneqq \conv \{ (a,\pm 1), (\pm 1,b) \} \subset \R^2 .
\end{equation*}
\end{definition}

\begin{definition}
The \emph{blunt square} is defined as $S = [-1,1]^2 \setminus \{-1,1\}^2 \subset \R^2$.
\end{definition}

Note that the cone $\CC(S)$ is not proper, since it is not closed. We rely on two propositions. The first shows that any non-classical cone can be `sandwiched' between cones based on a kite and on a blunt square. The second produces nontrivial information about the maximal tensor product of cones based on kites vs the minimal tensor product of cones based on blunt squares.

\begin{manualprop}{\ref{proposition:auerbach}}
Let $V$ be a $3$-dimensional vector space, and $\C \subset V$ a proper cone which is not classical. 
There exist $(a,b) \in (-1,1)^2$ and $\Phi : V \to \R^2 \times \R$ a linear bijection such that
\begin{equation*}
\CC(T_{a,b}) \subseteq \Phi(\C) \subseteq \CC(S)\, .
\end{equation*}
\end{manualprop}

\begin{manualprop}{\ref{proposition:sticks-out}}
Let $a_1$, $a_2$, $b_1$ and $b_2$ be elements of $(-1,1)$. Then
\begin{equation*}
\CC(T_{a_1,b_1}) \tmax \CC(T_{a_2,b_2}) \not\subseteq \CC(S) \tmin \CC(S).
\end{equation*}
In other words, there exists $\omega \in \CC(T_{a_1,b_1}) \tmax \CC(T_{a_2,b_2})$ such that $\omega \not\in \CC(S) \tmin \CC(S)$.
\end{manualprop}

We postpone the proof of Propositions~\ref{proposition:auerbach} and~\ref{proposition:sticks-out} to the end of this section, and show how they together imply Theorem~\ref{3-dim result}. The following reasoning was already sketched in the Methods section of the main article, but we repeat it here for the sake of completeness. Considering two 3-dimensional non-classical proper cones $\C_1$, $\C_2$, we need to show that the pair $(\C_1,\C_2)$ is entangleable. Assume by contradiction that $(\C_1,\C_2)$ is a nuclear pair. By combining Proposition~\ref{proposition:auerbach} with Fact~\ref{fact:nuclear-bijection}, we may assume that
\begin{equation*}
\CC(T_{a_1,b_1}) \subseteq \C_1 \subseteq \CC(S) \quad \textnormal{and}\quad \CC(T_{a_2,b_2}) \subseteq \C_2 \subseteq \CC(S)
\end{equation*}
for some $a_1, a_2, b_1, b_2\in (-1,1)$. Since $\tmin$ and $\tmax$ are increasing operations with respect to set inclusion, it follows that
\begin{equation*}
\CC(T_{a_1,b_1}) \tmax \CC(T_{a_2,b_2}) \subseteq \C_1 \tmax \C_2 = \C_1 \tmin \C_2 \subseteq \CC(S) \tmin \CC(S)\, ,
\end{equation*}
thus contradicting the conclusion of Proposition~\ref{proposition:sticks-out}.

\subsection{Proof of Proposition~\ref{proposition:auerbach}}


We prove the following statement, which implies Proposition~\ref{proposition:auerbach}. 
It is a variant of Auerbach's lemma, which is usually stated for symmetric convex bodies (see for example~\cite[II.E.11]{Wojtaszczyk}).

\begin{proposition} \label{proposition:auerbach2}
If $K \subset \R^2$ is a convex body which is not a triangle, then there exists an affine bijection $\Psi : \R^2 \to \R^2$ and $a,b \in (-1,1)$ such that $ T_{a,b} \subseteq \Psi(K) \subseteq S$.
\end{proposition}

To check that Proposition~\ref{proposition:auerbach2} implies Proposition~\ref{proposition:auerbach}, first note, using Fact~\ref{fact:cone-has-base}, that it is enough to prove Proposition~\ref{proposition:auerbach} for $\C = \CC(K)$, with $K$ a convex body in $\R^2$. Moreover, by Fact~\ref{fact:classical} we see that $\C$ is non-classical if and only if $K$ is not a triangle (i.e.\ a two-dimensional simplex). Then, Proposition~\ref{proposition:auerbach} follows from Proposition~\ref{proposition:auerbach2} together with Fact~\ref{fact:affine-linear}.

\begin{proof}[Proof of Proposition~\ref{proposition:auerbach2}]
Let $ABCD$ be a quadrilateral of maximal area inside $K$ (since $K$ is not a triangle, this quadrilateral does not degenerate into a triangle). The existence of this quadrilateral follows easily from a compactness argument. Basic geometric considerations (see Figure~\ref{figure-auerbach}) show that maximality implies that $K$ lies between the lines parallel to $(AC)$ passing through $B$ and $D$; and between the lines parallel to $(BD)$ passing through $A$ and $C$. These four lines delimit a parallelogram which can be mapped to $[-1,1]^2$ by a suitable affine bijection $\Psi$. At this step we showed the existence of $(a,b) \in (-1,1)^2$ such that 
$T_{a,b} \subseteq \Psi(K) \subseteq [-1,1]^2$. This is a bit weaker than the conclusion of the lemma.

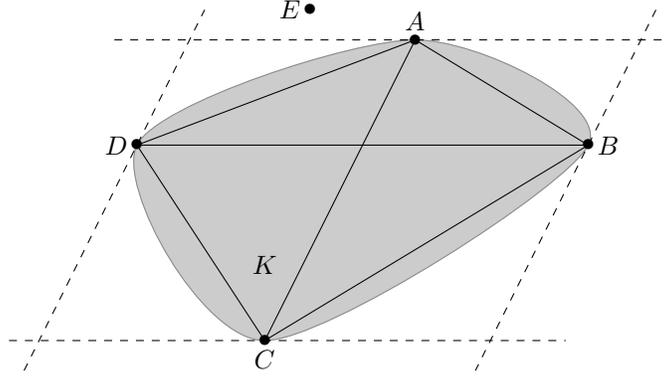
\begin{figure}[htbp] \begin{center}
\begin{tikzpicture}[scale=2]
\coordinate (A) at (0.5,1) ;
\coordinate (B) at (1.65,0.3) ;
\coordinate (C) at (-0.5,-1) ;
\coordinate (D) at (-1.35,0.3) ;
\coordinate (E) at (-0.2,1.2) ;

\draw [gray!40,fill] plot [smooth cycle] coordinates { (A) (B) (C) (D)};
\draw [gray!90] plot [smooth cycle] coordinates { (A) (B) (C) (D)};

\draw [dashed] (-1.5,1)--(2.2,1);
\draw [dashed] (-2.2,-1)--(1.5,-1);
\draw [dashed] (2.1,1.2)--(0.9,-1.2);
\draw [dashed] (-2.1,-1.2)--(-0.9,1.2);

\draw (A)--(B)--(C)--(D)--(A);
\draw (A)--(C);
\draw (B)--(D);

\draw (A) node {$\bullet$};
\draw (B) node {$\bullet$};
\draw (C) node {$\bullet$};
\draw (D) node {$\bullet$};

\draw (A) node[above] {$A$};
\draw (B) node[right] {$B$};
\draw (C) node[below] {$C$};
\draw (D) node[left] {$D$};

\draw (E) node {$\bullet$};
\draw (E) node[left] {$E$};

\draw (-0.5,-0.5) node {$K$};

\end{tikzpicture}
\end{center}
\caption{
The quadrilateral $ABCD$ has maximal area inside the convex body $K$ depicted in gray. It follows that $K$ lies in the parallelogram delimited by dotted lines: if for example a point $E$ is above the line parallel to $(BD)$ passing through $A$, then $\area(EBCD) > \area(ABCD)$ and therefore $E \not\in K$.
} \label{figure-auerbach}
\end{figure}

To enforce $\Psi(K) \subseteq S$ we need to be more careful in our construction. Among all quadrilaterals of maximal area inside $K$, choose $ABCD$ with the extra property that as few as possible among $A$, $B$, $C$ and $D$ are extreme points in $K$. We claim that repeating the construction from the previous paragraph with that choice of $ABCD$ implies that $\Psi(K) \subseteq S$. Indeed, since our construction is affine-invariant, we may assume that $\Psi$ is the identity, so that $T_{a,b} \subseteq K \subseteq [-1,1]^2$. The fact that $T_{a,b}$ is a quadrilateral of maximal area inside $K$ follows from the change of variables theorem, which ensures that the area of the image of a set $X$ by an affine transformation equals the area of $X$ times a constant.
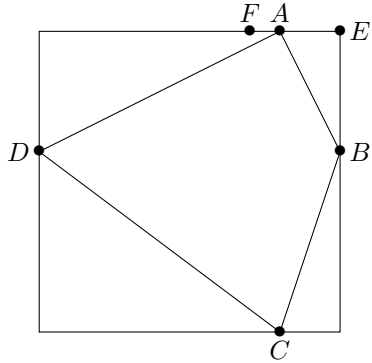
\begin{figure}[htbp] \begin{center}
\begin{tikzpicture}[scale=2]
	\coordinate (a) at (1,1);
	\coordinate (b) at (1,-1);
	\coordinate (c) at (-1,-1);
	\coordinate (d) at (-1,1) ;
	\coordinate (A) at (0.6,1);
	\coordinate (B) at (1,0.2);
	\coordinate (C) at (0.6,-1);
	\coordinate (D) at (-1,0.2) ;
	\draw (A)--(B)--(C)--(D)--(A);
	\draw (a)--(b)--(c)--(d)--(a);
	\draw (A) node {$\bullet$};
	\draw (A) node[above] {$A$};
    \draw (B) node {$\bullet$};
	\draw (B) node[right] {$B$};
    \draw (C) node {$\bullet$};
	\draw (C) node[below] {$C$};
    \draw (D) node {$\bullet$};
	\draw (D) node[left] {$D$};
    \draw (1,1) node {$\bullet$};
    \draw (1,1) node[right] {$E$};
    \draw (0.4,1) node {$\bullet$};
    \draw (0.4,1) node[above] {$F$};
    \end{tikzpicture} \end{center}
\caption{
If $E$ belongs $K$ and $A$ is not an extreme point of $K$, then $K$ contains the quadrilateral $FECD$ which has larger area than $ABCD$.
} \label{figure-auerbach-2}
\end{figure}
By symmetry it suffices to show that $E \coloneqq (1,1) \not\in K$. Suppose by contradiction that $E\in K$, and let consider the points $A=(a,1)$, $B=(1,b)$, $C=(a,-1)$ and $D=(-1,b)$. We claim that $A$ is not an extreme point of $K$. This follows from our choice of $ABCD$ together with the observation that, for every $X$ in the segment $AE$, we have $\area(ABCD)=\area(XBCD)$ while $X$ is not an extreme point of $K$ (see Figure~\ref{figure-auerbach-2}). It follows that there exists $a'<a$ such that the point $F \coloneqq (a',1)$ belongs to $K$. At this point we reached a contradiction: since $AECD$ is strictly contained in $FECD$, we have
\[ \area(FECD) > \area(AECD) = \area(ABCD). \]
Hence, we have found a quadrilateral $FECD$ inside $K$ which has an strictly larger area than $\area(T_{a,b})$.
\end{proof}

\subsection{Proof of Proposition~\ref{proposition:sticks-out}}

The cones $\CC(T_{a,b})$ and $\CC(S)$ live in $\R^2 \times \R$, which we identify with $\R^3$. The tensor products $\CC(T_{a_1,b_1}) \tmax \CC(T_{a_2,b_2})$ and $\CC(S) \tmin \CC(S)$ live in $\R^3 \otimes \R^3$, which we identify with the algebra $\gM_3$ of $3 \times 3$ matrices with real entries. We also use the canonical inner product on $\R^3 \otimes \R^3$, which allows to identify $(\R^3 \otimes \R^3)^*$ with $\R^3 \otimes \R^3$.

Given real parameters $a$ and $b$, consider the self-adjoint matrices
\[ M_{a,b} \coloneqq \begin{pmatrix} 1 & ab & a \\ ab & 1 & b \\ a & b & 1 \end{pmatrix}, \quad
H \coloneqq \begin{pmatrix} 1 & 1 & 0 \\ 1 & -1 & 0 \\ 0 & 0 & 1 \end{pmatrix}.
\]

The proof of Proposition~\ref{proposition:sticks-out} is completely explicit: given $a_1, b_1, a_2, b_2 \in (-1,1)$, we define a matrix $\Omega = (\omega_{i,j})$ by 
\[ \Omega \coloneqq M_{a_2,b_2} H^{-1} M_{a_1,b_1}. \]
We compute, using the notation $\alpha = a_1b_1$, $\beta=a_1b_2$, $\gamma=a_2b_1$, $\delta= a_2b_2$, 
$\e=a_1a_2$, $\zeta=b_1b_2$, $\eta=a_1a_2b_1b_2$,
\begingroup 
\setlength\arraycolsep{10pt}
\[ \Omega = \frac{1}{2} \begin{pmatrix} 1 + \alpha + \delta + 2\e - \eta &  1 + \alpha + 2\gamma -\delta + \eta & * \\ 1-\alpha + 2\beta + \delta +\eta & -1 + \alpha + \delta + 2\zeta + \eta & * \\ * & * &  2 + \beta + \gamma +\e -\zeta \end{pmatrix} ,\]
\endgroup
where entries whose values are not used in our computation are denoted by $*$. We check in particular that
\begin{equation} \label{eq:omega} \omega_{11} + \omega_{12} + \omega_{21} - \omega_{22} = 2 \omega_{33}. \end{equation}


Let $\omega \in \R^3 \otimes \R^3$ be the tensor which is identified with $\Omega \in \gM_3$. We claim
that
\begin{equation} \label{eq:omega-1} \omega \in \CC(T_{a_1,b_1}) \tmax \CC(T_{a_2,b_2}), \end{equation}
\begin{equation} \label{eq:omega-2} \omega \not\in \CC(S) \tmin \CC(S), \end{equation}
and the proof of Proposition~\ref{proposition:sticks-out} will be complete provided we justify~\eqref{eq:omega-1} and~\eqref{eq:omega-2}.

\begin{proof}[Proof of~\eqref{eq:omega-1}]
Denote by $K=\conv\{(\pm 1,0),(0,\pm 1)\} \subset \R^2$. 
\begin{fact} 
The matrices $H$ and $M_{a,b}$ have the following properties
\begin{enumerate}[label={S\arabic{definition}.\arabic*}, leftmargin=\widthof{S34.2}+\labelsep]
\item \label{fact:H-Mab1}  For every $a,b$ in $(-1,1)$, we have $M_{a,b}(\CC(K)) = \CC(T_{a,b})$.
\item \label{fact:H-Mab2}  We have $H(\CC(K))=\CC(K)^*$.
\end{enumerate}
\end{fact}

\begin{proof}
For the first part, note that $M_{a,b}(\pm 1, 0, 1) = (1 \pm a)(\pm 1,b,1)$ and $M_{a,b}(0,\pm 1,1) = (1\pm b)(a,\pm 1,1)$, so that $M_{a,b}$ maps extreme rays of $\CC(K)$ to extreme rays of $\CC(T_{a,b})$.
For the second part, we check that $\CC(K)^* = \CC([-1,1]^2)$, and that $H$ maps extreme rays of $\CC(K)$
to extreme rays of $\CC([-1,1]^2)$.
\end{proof}

In view of Fact~\ref{fact:tmax-and-pos},~\eqref{eq:omega-1} is equivalent to
the fact that $\Omega = \mathrm{op}(\omega) \in \gPos(\CC(T_{a_1,b_1})^*,
\CC(T_{a_2,b_2}))$, or again to the inclusion
\[ M_{a_2,b_2} H^{-1} M_{a_1,b_1} ( \CC(T_{a_1,b_1})^* ) \subseteq \CC(T_{a_2,b_2}) \]
or further (using Fact~\ref{fact:H-Mab1}) that
\[ M_{a_2,b_2} H^{-1} M_{a_1,b_1} \big( ( M_{a_1,b_1} \CC(K) )^* \big) \subseteq M_{a_2,b_2} \CC(K). \]
We now invoke Fact~\ref{fact:polar-linearmap} (applied with $\Phi=\Phi^*=M_{a_1,b_1})$ to claim that this statement follows from the inclusion $H^{-1}(\CC(K)^*) \subseteq \CC(K)$, which is an obvious consequence of Fact~\ref{fact:H-Mab2}.
\end{proof}

\begin{proof}[Proof of~\eqref{eq:omega-2}]

Consider $a \in \CC(S) \tmin \CC(S)$, identified with a matrix $(a_{ij}) \in \gM_3$. We claim that it satisfies the inequality
\begin{equation} \label{eq:two-blunt-squares}
a_{11} + a_{12} + a_{21} - a_{22} < 2 a_{33}.
\end{equation}
Once~\eqref{eq:two-blunt-squares} is proved,~\eqref{eq:omega-2} follows immediately by comparison with~\eqref{eq:omega}. To show~\eqref{eq:two-blunt-squares}, we use the following variant of the CHSH inequality. 

\begin{lemma}[CHSH inequality, strict version] \label{lemma:CHSH}
If $(x_1,y_1) \in S$ and $(x_2,y_2) \in S$, then
\begin{equation} \label{eq:CHSH-strict} x_1x_2 + x_1 y_2 + y_1x_2 - y_1y_2 < 2 .\end{equation}
\end{lemma}

\begin{proof}[Proof of Lemma~\ref{lemma:CHSH}] We have
\begin{align*} | x_1x_2 + x_1 y_2 + y_1x_2 - y_1y_2 | & \leq |x_1| \cdot |x_2+y_2| + |y_1| \cdot |x_2-y_2| \\
 & \leq |x_2+y_2| + |x_2-y_2| \\
 & \leq 2.
\end{align*}
We argue that one of the inequalities must be strict. Assume the last inequality to be an equality. In this case, either $|x_2|$ or $|y_2|$ must equal $1$. Since they cannot both equal $1$, the numbers $x_2+y_2$ and $x_2-y_2$ are nonzero. Now, if the second inequality is also an equality, then it follows that $|x_1|=|y_1|=1$, a contradiction.
\end{proof}

To complete the proof of~\eqref{eq:two-blunt-squares}, note that 
any nonzero element $z \in \CC(S) \tmin \CC(S)$ is a positive combination of elements of the form  $(x_1,y_1,1) \otimes (x_2,y_2,1)$ with $(x_1,y_1)$ and $(x_2,y_2)$ in $S$. It is enough to test~\eqref{eq:two-blunt-squares} on elements of that form, in which case it reduces to~\eqref{eq:CHSH-strict}.
\end{proof}

\section{Proof of Theorem~\ref{polyhedral result}: polyhedral cones} \label{proof:polyhedral}

The proof of Theorem~\ref{polyhedral result} is based on the following proposition, whose proof is postponed to the end of the section.

\begin{proposition}[non-classical polyhedral cones have non-classical retracts]
\label{proposition:polyhedral-retracts}
Let $\C$ be a proper polyhedral cone which is non-classical. Then there is a non-classical $3$-dimensional proper polyhedral cone $\C'$ which is a retract of $\C$. 
\end{proposition}

To prove Theorem~\ref{polyhedral result}, consider $\C_1$, $\C_2$ proper polyhedral cones, and suppose that $(\C_1,\C_2)$ is nuclear. Assume by contradiction that $\C_1$ and $\C_2$ are non-classical. By Proposition~\ref{proposition:polyhedral-retracts}, there exist non-classical $3$-dimensional proper polyhedral cones $\C'_1$ and $\C'_2$ which are retracts of $\C_1$ and $\C_2$, respectively. By Proposition~\ref{proposition:retracts-nuclear}, it follows that the pair $(\C'_1,\C'_2)$ is nuclear. This contradicts Theorem~\ref{3-dim result}.

\begin{proof}[Proof of Proposition~\ref{proposition:polyhedral-retracts}]

We rely on the following two lemmas.

\begin{manuallemma}{\ref{lemma:facet-retract}}[see also Exercise 2.18 in~\cite{ziegler}]
Let $F$ be a facet of a proper polyhedral cone $\C$. Then $F$, which is a proper polyhedral cone when seen as a subset of $\mathspan(F)$, is a retract of $\C$.
\end{manuallemma}

\begin{manuallemma}{\ref{lemma:simple-simplicial}}
Let $\C$ be a non-classical proper polyhedral cone with $\dim (\C) \geq 4$. Then either $\C$ or $\C^*$ has a facet which is non-classical.
\end{manuallemma}

We now prove Proposition~\ref{proposition:polyhedral-retracts} by induction on the dimension. Let $\C$ be a non-classical proper polyhedral cone of dimension $n$ (note that $n \geq 3$, since any $2$-dimensional cone is classical). If $n=3$, Proposition~\ref{proposition:polyhedral-retracts} is obviously true. If $n \geq 4$,  using Lemmas~\ref{lemma:facet-retract} and~\ref{lemma:simple-simplicial}, we obtain that either (i) $\C$ has a non-classical polyhedral retract of dimension $n-1$ or (ii) $\C^*$ has a non-classical polyhedral retract of dimension $n-1$. Using Facts~\ref{fact:dual-classical} and~\ref{fact:retract-dual}, we check that (i) and (ii) are in fact equivalent, so (i) always holds. Since a retract of a retract of $\C$ is also a retract of $\C$, Proposition 
\ref{proposition:polyhedral-retracts} follows by induction.
\end{proof}

For the proof of Lemma~\ref{lemma:facet-retract} we will use the following standard fact (see {\cite[Theorem 5.8]{Brondsted83}}).

\begin{fact} \label{fact:facets-are-exposed}
Let $C$ be a closed convex set in a vector space $V$, and $F$ be a facet of $C$. Then there is a linear form $f \in V^*$ and a real number $a$ such that $f(x) \geq a$ for every $x \in C$, and $C \cap f^{-1}(a) = F$.
\end{fact}


\begin{proof}[Proof of Lemma~\ref{lemma:facet-retract}]
Let $W = \aff(F)$; note that $\dim W = \dim V -1$. By Fact~\ref{fact:facets-are-exposed} and the homogeneity of $\C$ (used to enforce $a=0$), there is a linear form $f \in V^*$ such that $f \geq 0$ on $\C$, $\ker f = W$ and $W \cap \C =F$. 

Let $\pi : V \to W$ be a projection onto $W$. Pick an element $x$ in the relative interior of $F$, and define for $\lambda>0$ a linear map $\Phi_\lambda : V \to W$ by $\Phi_\lambda(z) = \pi(z) + \lambda f(z)x$ for $z \in V$. If $z\in F\subset W$, then $\Phi_\lambda(z)=\pi(z)=z\in F$. On the other hand, for every $z \in \C \setminus F$, we have
\[  \Phi_\lambda(z) = \lambda \Big( f(z)x + \lambda^{-1} \pi(z) \Big). \]
Since $f(z)>0$, the point $f(z)x$ is in the relative interior of $F$, and we have $\Phi_\lambda(z) \in F$ for $\lambda$ larger than some number $\lambda_0(z)>0$. Since $\C = \cone(A)$ for some finite set $A$, it follows that $\Phi_\lambda(\C) \subset F$ for $\lambda$ larger than $\max \{ \lambda_0(z):\, z\in A \}$.
If $\Psi : W \to V$ denotes the canonical inclusion, $\Phi_\lambda \circ \Psi$ is the identity on $W$, and this shows that $F$ is a retract of $\C$.
%
%
\end{proof}

Lemma~\ref{lemma:simple-simplicial} is a reformulation for polyhedral cones of a basic result on polytopes (Fact~\ref{fact:simple-simplicial}). A \emph{polytope} is a convex body which is the convex hull of a finite set. A $d$-dimensional polytope is said to be \emph{simplicial} if all its facets contain exactly $d$ vertices (i.e.\ they are $(d-1)$-dimensional simplices), and is said to be \emph{simple} if every vertex belongs to $d$ facets. 

\begin{fact}[See {\cite[Theorem 12.19]{Brondsted83}}]
 \label{fact:simple-simplicial}
For $d \geq 3$, a $d$-dimensional polytope which is both simple and simplicial is a simplex.
\end{fact}

We also use basic properties of the duality of polytopes, which is defined for example in~\cite[\S 10]{Brondsted83}.

\begin{fact}[See {\cite[Theorem 12.10]{Brondsted83}}] \label{fact:dual-simple-simplicial}
Let $P$ and $Q$ be dual polytopes. Then $P$ is simple if and only if $Q$ is simplicial.
\end{fact}

\begin{fact}[See {\cite[Lemma 1.6]{ABMB}}] \label{fact:dual-cone-dual-polytope}
Let $P$ be a polytope. Then there is a polytope $Q$ dual to $P$ such that the cones $\CC(P)^*$
and $\CC(Q)$ are isomorphic.
\end{fact}

\begin{proof}[Proof of Lemma~\ref{lemma:simple-simplicial}]
Let $\C$ be a non-classical proper polyhedral cone with $\dim (\C) \geq 4$. Without loss of generality (Fact~\ref{fact:cone-has-base}) we may assume that $\C=\CC(P)$ for some polytope $P$ of dimension at least $3$. Since $\C$ is not classical, $P$ is not a simplex (Fact~\ref{fact:classical}). By Fact~\ref{fact:simple-simplicial}, $P$ cannot be both simple and simplicial.
\begin{itemize}
\item If $P$ is not simplicial, then one of its facets is not a simplex, and therefore one of the facets of $\CC(P)$ is not classical (see Fact~\ref{fact:faces-cone-base}). 
\item It $P$ is not simple, let $Q$ be the polytope dual to $P$ given by Fact~\ref{fact:dual-cone-dual-polytope}. By Fact~\ref{fact:dual-simple-simplicial}, $Q$ is not simplicial. Repeating the reasoning above with $Q$ instead of $P$ and using Fact~\ref{fact:dual-cone-dual-polytope} shows that one of the facets of $\CC(P)^*$ is not classical.
\end{itemize}
In both cases, the conclusion of Lemma~\ref{lemma:simple-simplicial} is verified.
\end{proof}

\section{Proof of Theorem~\ref{semiquantum result}: positive semidefinite cones} \label{section:semiquantum}

We start with a simple observation about the positive semidefinite cones introduced in Definition~\ref{def:PSD}.

\begin{fact} \label{fact:pinching-retract}
If $k \leq n$, then $\PSD_k$ is a retract of $\PSD_n$.
\end{fact}

\begin{proof}
Identify $\mathbf{C}^k$ with a subspace of $\mathbf{C}^n$, let $\pi : \mathbf{C}^n \to \mathbf{C}^k$ a projection onto $\mathbf{C}^k$. Then the identity map $\Psi : \gM_k \to \gM_n$ and the map $\Phi : \gM_n \to \gM_k$ defined by $\Phi(A) = \pi A \pi^\dagger$ are positive and satisfy $\Phi \circ \Psi = \Id_{\gM_k}$, as needed.
\end{proof}

Our proof of Theorem~\ref{semiquantum result} is based on properties of the \emph{Lorentz cone}, defined for an integer $n$ by
\[ \mathsf{L}_n \coloneqq \left\{ (x_1,\dots,x_{n+1}) \in \R^{n+1} \, : \, \sqrt{x_1^2 + \cdots + x_n^2} \leq x_{n+1} \right\} .\]
Note that $\dim(\mathsf{L}_{n}) = n+1$, and that the cone $\mathsf{L}_n$ is self-dual, i.e.\ $\mathsf{L}_n^*=\mathsf{L}_n$ (we identify $\R^n$ with its dual space in the usual way). Theorem~\ref{semiquantum result} will be a consequence of the following propositions.

\begin{proposition}\label{proposition:ice-cream}
Let $\C$ be a proper cone with $\dim (\C) \leq n+1$. Then the pair $(\C,\mathsf{L}_n)$ is nuclear if and only if $\C$ is classical.
\end{proposition}

\begin{proposition} \label{proposition:PSD-Lorentz}
For every $n \geq 1$, the cone $\mathsf{L}_{2n}$ is a retract of $\PSD_{2^n}$.
\end{proposition}

Before proving Propositions~\ref{proposition:ice-cream} and~\ref{proposition:PSD-Lorentz}, we show how they imply together Theorem~\ref{semiquantum result}. The easy direction has been covered in the main text. For the other direction, assume that the pair $(\C,\PSD_n)$ is nuclear, with $\lfloor \log_2 n \rfloor \geq d/2$ and $\dim (\C)=d$. Using Fact~\ref{fact:pinching-retract} and Proposition~\ref{proposition:PSD-Lorentz} together with our assumption $\lfloor \log_2 n\rfloor \geq \frac{d-1}{2}$ implies that $\mathsf{L}_{d-1}$ is a retract of $\PSD_n$. By Proposition~\ref{proposition:retracts-nuclear}, we obtain that the pair $(\C,\mathsf{L}_{d-1})$ is nuclear. Finally, Proposition~\ref{proposition:ice-cream} implies that $\C$ is classical.

Let us first consider for $r>0$
\[ \mathsf{L}_n(r) \coloneqq \left\{ (x_1,\dots,x_{n+1}) \in \R^{n+1} \, : \, \sqrt{x_1^2 + \cdots + x_n^2} \leq rx_{n+1} \right\} \] and note that $\mathsf{L}_n(r)=\mathscr{C}(rB_2^n)$, where $B_2^n$ is the unit Euclidean ball in $\R^n$. Obviously, $\mathsf{L}_n=\mathsf{L}_n(1)$.

In order to prove Proposition~\ref{proposition:ice-cream}, we will need the following lemma. It is a consequence of~\cite[Lemma 19]{ultimate}. However, for completeness, we include a direct proof.
\begin{lemma}\label{min-max ice cream}
Given a natural number $n$ and a nonnegative real number $r$, the inclusion $$\mathsf{L}_n\tmax \mathsf{L}_n\subseteq \mathsf{L}_n\tmin \mathsf{L}_n(r)$$ implies $r\geq n$.
\end{lemma}

We will also use the following extremal properties of simplices, which can be found in~\cite{Leichtweiss} (see also~\cite[Theorem 1]{Palmon}).
\begin{theorem}\label{theorem:leichtweiss}
Let $K \subset \R^n$ be a convex body which is not a simplex. Then there exists an affine bijection $\Phi : \R^n \to \R^n$ such that $B_2^n \subseteq \Phi(K) \subseteq r  B_2^n$ for some positive number $r$ satisfying $r <n$.
\end{theorem}

Applying Theorem~\ref{theorem:leichtweiss} to $K$ being the base of a cone, and using Fact~\ref{fact:affine-linear}, we obtain the following variant.

\begin{corollary}\label{corollary:leichtweiss}
Let $\C \subset \R^{n+1}$ be a proper cone which is not classical. Then $\C$ is isomorphic to a proper cone $\C' \subset \R^{n+1}$ satisfying $\mathsf{L}_n \subseteq \C' \subseteq \mathsf{L}_n(r)$ for some number $r$ with $1 \leq r<n$.
\end{corollary}

\begin{proof}[Proof of Proposition~\ref{proposition:ice-cream}]
We will proceed by contradiction. Assume that there exists a non-classical proper cone $\C$ with $\dim (\C) = n+1$ and such that
the pair $(\C,\mathsf{L}_n)$ is nuclear. Using Corollary~\ref{corollary:leichtweiss} (and Fact~\ref{fact:nuclear-bijection}), we may assume that $\mathsf{L}_n\subseteq \mathcal C\subseteq \mathsf{L}_n(r)$ for some $r<n$. Then, we can write 
\[ \mathsf{L}_n\tmax\mathsf{L}_n\subseteq \mathsf{L}_n\tmax\mathcal C= \mathsf{L}_n\tmin\mathcal C\subseteq \mathsf{L}_n\tmin\mathsf{L}_n(r), \] 
contradicting Lemma~\ref{min-max ice cream}.
\end{proof}


\begin{proof}[Proof of Lemma~\ref{min-max ice cream}]
Let us consider the element $z=\sum_{i=1}^{n+1}e_i\otimes e_i\in \R^{n+1}\otimes \R^{n+1}$, with $(e_i)$ the canonical basis of $\R^{n+1}$. We will show that $z\in \mathsf{L}_n\tmax\mathsf{L}_n$, while $z\in \mathsf{L}_n\tmin\mathsf{L}_n(r)$ implies $r\geq n$.

Let us first show that $z\in \mathsf{L}_n\tmax\mathsf{L}_n=(\mathsf{L}_n^*\tmin\mathsf{L}_n^*)^*=(\mathsf{L}_n\tmin\mathsf{L}_n)^*$, where in the last equality we have used that $\mathsf{L}_n$ is a selfdual cone. To this end, it is enough to check the inequality $\langle z, a \otimes b \rangle \geq 0$ for $a$, $b \in \mathsf{L}_n$, since elements of the form $a \otimes b$ generate the cone $\mathsf{L_n} \tmin \mathsf{L}_n$. We then write 
\[
\langle z, a \otimes b \rangle  = \sumno_{i=1}^{n+1} a_ib_i \geq a_{n+1}b_{n+1}-\left|\sumno_{i=1}^{n}a_ib_i\right| \geq a_{n+1}b_{n+1}-\left(\sumno_{i=1}^{n}a_i^2\right)^{\frac{1}{2}}\left(\sumno_{i=1}^{n}b_i^2\right)^{\frac{1}{2}},
\]
where in the last step we have used Cauchy--Schwarz inequality. The last expression is nonnegative because $a$, $b\in \mathsf{L}_n$, and we conclude that $z\in (\mathsf{L}_n\tmin\mathsf{L}_n)^*=\mathsf{L}_n\tmax\mathsf{L}_n$.

In order to show that $z=\sum_{i=1}^{n+1}e_i\otimes e_i\in \mathsf{L}_n\tmin\mathsf{L}_n(r)$ implies $r\geq n$, let us consider an arbitrary decomposition $z=\sum_k x_k\otimes y_k$ such that $x_k\in \mathsf{L}_n$ and $y_k\in \mathsf{L}_n(r)$ for every $k$. We denote by $\|\cdot\|_2$ the standard Euclidean norm on $\R^{n+1}$ and by $\|\cdot \|_1$ the trace norm in $\R^{n+1}\otimes \R^{n+1}$ identified with $\gM_{n+1}(\R)$. We have 
\begin{align*}
n &= \left\|\sumno_{i=1}^{n}e_i\otimes e_i\right\|_1 \\
&= \left\|\sumno_k (x_k(i))_{i=1}^{n}\otimes (y_k(i))_{i=1}^{n}\right\|_1 \\
&\leq \sum_k \left\|(x_k(i))_{i=1}^{n}\otimes (y_k(i))_{i=1}^{n}\right\|_1 \\
&= \sum_k \left\|(x_k(i))_{i=1}^{n}\right\|_2 \left\|(y_k(i))_{i=1}^{n}\right\|_2 \\
&\leq r \sum_k x_k(n+1)y_k(n+1)\\
&= r.
\end{align*}
This concludes the proof.
\end{proof}

\begin{proof}[Proof of Proposition~\ref{proposition:PSD-Lorentz}]
We use the well-known fact (see for example the proof of Lemma 11.2 in~\cite{ABMB}) that one can find self-adjoint and trace zero matrices $U_1,\dots , U_{2n}\in \gM_{2^n}^{\mathrm{sa}}(\mathbb C)$ such that $$U_iU_j+U_jU_i=2\delta_{i,j}\text{Id} \text{    }\text{   for every  }i,j=1,\dots, 2n.$$
This property immediately implies that for all real numbers $x_1,\dots, x_{2n}$, 
\begin{equation} \label{PSD}
\left(\sumno_{i=1}^{2n}x_iU_i\right)^2 = \left(\sumno_{i=1}^{2n}x_i^2\right)\text{Id}.
\end{equation}
Define linear maps $\Phi: \gM_{2^n}^{\mathrm{sa}} \to  \R^{2n+1}$ and $\Psi:  \R^{2n+1}\to \gM_{2^n}^{\mathrm{sa}}$ by 
\[ \Phi(A)=\left(\tr (AU_1),\cdots, \tr (AU_{2n}), \tr  A \right) \]
and
\[\Psi(x)=\sum_{i=1}^{2n}x_iU_i+x_{2n+1}\Id\, .\] 
Proposition~\ref{proposition:PSD-Lorentz} is an immediate consequence of the following three properties:
\begin{enumerate}[(1)]
\item $\Phi(\PSD_n)\subset \mathsf{L}_{2n}$;
\item $\Psi(\mathsf{L}_{2n})\subset \PSD_n$; and
\item $\Phi\circ \Psi=2^n \text{Id}_{\R^{2n+1}}$.
\end{enumerate}

In order to prove (1), note that for every $A\in \PSD_{2^n}$ and $x \in \R^{2n}$, we have that
\[ \sum_{i=1}^{2n}x_i \tr (AU_i) = \tr \left(A\sumno_{i=1}^{2n}x_i U_i\right)\leq \|A\|_1 \left\|\sumno_{i=1}^{2n}x_i U_i\right\|= \|x\|_2 \tr A\, . \]
By taking the supremum over $x$ such that $\|x\|_2\leq 1$, we conclude that 
\[ \left(\sumno_{i=1}^{2n}(\tr AU_i)^2\right)^{1/2}\leq \tr A\]
and therefore $\Phi(A)\in \mathsf{L}_{2n}$.

To prove (2), observe that for a given $x\in \mathsf{L}_{2n}$ the matrix $\Psi(x)=\sum_{i=1}^{2n}x_iU_i+x_{2n+1}\text{Id}$ is positive semidefinite due to~\eqref{PSD}.
Finally, (3) follows from the facts that $\tr (U_iU_j)=2^n \delta_{i,j}$ and $\tr (U_i)=0$ for every $i$.
\end{proof}

The tools we introduced can be used to derive a simple proof of Theorem~\ref{semiquantum result}'.
\begin{proof}[Proof of Theorem~\ref{semiquantum result}']
Let $\C$ be a proper polyhedral cone which is not classical, and $n \geq 2$. We combine the following facts: (a) by Proposition~\ref{proposition:polyhedral-retracts}, $\C$ admits a non-classical retract $\C'$ of dimension $3$, and (b) the Lorentz cone $\mathsf{L}_2$ is a retract of $\PSD_n$. Since $\mathsf{L}_2$ is not classical (a disk is not a triangle!), the pair $(\C',\mathsf{L}_2)$ is entangleable (Theorem~\ref{3-dim result}) and therefore the pair $(\C,\PSD_n)$ is entangleable as well (Proposition~\ref{proposition:retracts-nuclear}).
It remains to justify point (b) in the previous argument. This is easy: by Fact~\ref{fact:pinching-retract}, it is enough to prove (b) for $n=2$. Since $\PSD_2$ is isomorphic to $\mathsf{L}_3$, this amounts to proving that $\mathsf{L}_2$ is a retract of $\mathsf{L}_3$, which is geometrically obvious.
\end{proof}

\section{Proof of Theorem~\ref{symmetric result}}

\subsection{Entanglement robustness in GPTs}

The entanglement robustness is an entanglement measure that was constructed and studied in the early days of entanglement theory~\cite{VidalTarrach}. It differs from other quantifiers in that it has a purely geometric nature. In fact, it can be thought of as the minimal amount of noise, in the form of a convex mixture with a separable state, that makes the state separable. As it is rooted in convex geometry alone, the entanglement robustness can be extended from quantum mechanics to arbitrary GPTs in a straightforward manner. Given two arbitrary GPTs $(V_1,C_1,u_1)$ and $(V_2,C_2,u_2)$, and a state $\omega \in C_1\tmax C_2$, we set 
\begin{equation} \tag{\ref{ent-rob}}
\erob (\omega) \coloneqq \min\left\{ (u_1\otimes u_2)(\zeta):\, \zeta,\, \omega+\zeta\in C_1\tmin C_2 \right\} .
\end{equation}

It is not difficult to verify that the above definition possesses all the basic properties of an entanglement measure. First, it is everywhere non-negative and finite, because $C_1\tmin C_2$ is a proper cone with a non-empty interior. Secondly, it is faithful, namely, it vanishes (only on) separable states. Lastly, it is monotonically non-increasing under normalised separability-preserving maps, as the next lemma shows.

\begin{lemma} \label{lemma:ent-rob-monotone}
For $i=1,2$, let $(V_i,C_i,u_i)$ and $(V_i',C_i',u_i')$ be GPTs. Consider two composites $C_{12}$ and $C'_{12}$ such that~\eqref{double bound} holds for both. Let $\Lambda: V_1\otimes V_2\to V'_1\otimes V'_2$ be a map that is: (i)~positive, i.e.\ obeys $\Lambda(C_{12})\subseteq C'_{12}$; (ii)~normalised, i.e.\ satisfies $\Lambda^*(u'_1\otimes u'_2) = u_1\otimes u_2$; and (iii)~separability-preserving, namely, such that $\Lambda\left(C_1\tmin C_2\right) \subseteq C'_1\tmin C'_2$. Then, for all input states $\omega \in C_{12}$ it holds that $\erob(\Lambda(\omega))\leq \erob(\omega)$.
\end{lemma}

\begin{proof}
Let $\zeta\in C_1 \tmin C_2$ be the vector that achieves the minimum in~\eqref{ent-rob} for $\omega$. Since $\Lambda(\zeta), \Lambda(\omega)+\Lambda(\zeta)\in C'_1\tmin C'_2$, we have that
\begin{equation*}
    \erob(\Lambda(\omega)) \leq (u'_1\otimes u'_2)\left( \Lambda(\zeta)\right) = \left( \Lambda^*(u'_1\otimes u'_2)\right)(\zeta) = (u_1\otimes u_2)(\zeta) = \erob(\omega)\, ,
\end{equation*}
which completes the proof.
\end{proof}

\subsection{Symmetric cones}

We start by fixing some notation. For a proper cone $(V,C,u)$, define the vector subspace
\begin{equation} \label{X-space}
    X \coloneqq \ker (u) \coloneqq \{x\in V:\, u(x)=0\}\, ;
\end{equation}
once we fix a state $\gamma\in \Omega\coloneqq C\cap u^{-1}(1)$, we can decompose $V=X\oplus (\R \gamma)$, meaning that every $v\in V$ can be written as $v=\alpha \gamma+x$, where $\alpha\in\R$ and $x\in X$. 

We now construct the function $N_X:X\to \R_+$ given by
\begin{equation}
    N_X(x) \coloneqq \inf \left\{t>0:\, \gamma + t^{-1} x\in C \right\} ,
    \label{NX}
\end{equation}
where it is understood that the infimum of the empty set is $+\infty$. Note that in~\eqref{NX} we can substitute $C$ with $\Omega$.
Now, $N_X$ obeys the triangle inequality in general, i.e.\ that $N_X(x+y)\leq N_X(x)+N_X(y)$ for all $x,y\in X$. Also, $N_X$ is manifestly positively homogeneous, i.e.\ $N_X(\lambda x)=\lambda N_X(x)$ for all $\lambda\geq 0$. If $\Omega$ is centrally symmetric with centre $\gamma$, then also absolute homogeneity holds, i.e.\ $N_X(\lambda x)=\lambda N_X(x)$ for all $\lambda\in \R$.

As already mentioned, it follows from Fact~\ref{fact:cone-has-base}(c) that the state space $\Omega$ is a convex body when viewed as a subset of the affine space $u^{-1}(1)=\gamma+X$. This implies that the function $N_X$ defined by~\eqref{NX} satisfies $N_X(x)>0$ for all $x\neq 0$. Indeed, if this were not the case we will immediately deduce that $\gamma + n x\in \Omega$ for all positive integers $n$, which would in turn imply that $\Omega$ is non-compact.
If $\gamma\in \relint(\Omega)$, then we also have that $N_X(x)<\infty$ for all $x\in X$. Observe that the centre of a convex body, if it exists, must lie in its relative interior. We summarise the above discussion as follows.

\begin{fact} \label{fact:symmetric-gauge-norm}
For any symmetric GPT $(V,C,u)$, the function $N_X$ defined by~\eqref{NX} is a norm on $X$.
\end{fact}

From now on we will assume that $(V,C,u)$ is a symmetric GPT. Accordingly, we will adopt the more suggestive notation $\|\cdot\|_X\coloneqq N_X(\cdot)$. Observe that the unit ball of $\|\cdot\|_X$ is just the state space of $(V,C,u)$, i.e.
\begin{equation}
    B_X = \Omega - \gamma = C \cap u^{-1}(1) - \gamma\, . 
\end{equation}
For future convenience, let us also define $\Pi:V\to X$ as the projection onto $X$ with $\gamma$ in its kernel, explicitly given by the formula 
\begin{equation} \label{projection pi}
\Pi(v)\coloneqq v - u(v) \gamma\, ,
\end{equation}
for all $v\in V$. Observe that 
\begin{equation} \label{symmetric cone reconstruction}
C = \left\{ v\in V:\, \left\|\Pi(v)\right\|_X \leq u(v) \right\} .
\end{equation}
Note that since norms are invariant under a change of sign, we have that $u(v)\gamma - \Pi(v)\in C$ for all $v\in C$; in fact, 
\begin{equation}
\left\|\Pi(u(v)\gamma - \Pi(v))\right\|_X = \left\|-\Pi(v)\right\|_X \leq u(v) = u\left(u(v)\gamma - \Pi(v)\right) .
\end{equation}
We formalise this observation as follows.

\begin{fact} \label{fact:inversion-invariance}
A symmetric cone $C$ is invariant under the inversion $v\mapsto u(v)\gamma - \Pi(v) = 2u(v)\gamma - v$.
\end{fact}

Another curious feature of symmetric GPTs is that their dual spaces can also be equipped with a GPT structure.

\begin{lemma} \label{lemma:symmetric GPTs dualise}
Let $(V,C,u)$ be a symmetric GPT with centre $\gamma\in V=V^{**}$. Then
\begin{enumerate}[(a)]
    \item $(V^*,C^*,\gamma)$ is also a symmetric GPT with centre $u$;
    \item the corresponding space $X^*\coloneqq \ker(\gamma)$ normed by $B_{X^*}\coloneqq C^* \cap \gamma^{-1}(1) - u$ coincides (as a normed space) with the dual of $X$; and
    \item the projection onto $X^*$ with $u$ in its kernel coincides with the adjoint of $\Pi$ defined by~\eqref{projection pi}.
\end{enumerate}
\end{lemma}

\begin{proof}
Let $v^*\in \Omega_*\coloneqq C^*\cap \gamma^{-1}(1)$, and let us show that $2u-v^*\in \Omega_*$. On the one hand, clearly $(2u-v^*)(\gamma) = 1$. On the other, for an arbitrary $v \in C$ satisfying the constraint in~\eqref{symmetric cone reconstruction}, we see that $(2u-v^*)(v) = v^*(2 u(v) \gamma - v)\geq 0$, where the last inequality follows from Fact~\ref{fact:inversion-invariance}. We conclude that $2u-v^*\in C^*$, proving the first claim. As for the second, pick $f\in \ker(\gamma)\subset V^*$; then
\begin{align*}
    \inf \left\{t>0:\, u + t^{-1} f\in C^* \right\} &= \inf \left\{t>0:\, (u + t^{-1} f)(v)\geq 0\quad \forall\, v\in C \right\} \\
    &= \inf \left\{t>0:\, (u + t^{-1} f)(\gamma - x)\geq 0\quad \forall\, x\in B_{X} \right\} \\
    &= \inf \left\{t>0:\, f(x) \leq t \quad \forall\, x\in B_{X} \right\} \\
    &= \sup_{x\in B_{X}} f(x) \\
    &= \|f\|_{X^*}\, .
\end{align*}
This proves that the norm induced on $X^*$ by the construction in~\eqref{NX} coincides with the dual norm of $\|\cdot\|_X$ as given by Definition~\ref{def:dual-norm}.
The third claim is also straightforward. It suffices to observe that for all $v\in V$ and $v^*\in V^*$ it holds that
\begin{equation*}
    (v^* - v^*(\gamma) u)(v) = v^*(v) - u(v) v^*(\gamma) = v^*(v - u(v)\gamma) = v^*\left(\Pi (v)\right) = \left(\Pi^*(v^*)\right) (v)\, ,
\end{equation*}
implying that $\Pi^*(v^*) = v^* - v^*(\gamma) u$.
\end{proof}


We now move on to the bipartite setting. For $i=1,2$, let $(V_i,C_i,u_i)$ be a symmetric GPT with centre $\gamma_i\in \Omega_i\coloneqq C_i \cap u_i^{-1}(1)$. Call $X_i\coloneqq \ker(u_i)$, and let $\Pi_i$ be the projection onto $X_i$ with $\gamma_i\in \ker(\Pi_i)$.

\begin{lemma}[{see~\cite[Proposition~2.25]{lamiatesi}}] \label{lemma:norms-cones}
With the above notation, for every $\omega\in V_1\otimes V_2$ we have that
\begin{enumerate}
\item[(a)] if $\omega\in C_1\tmin C_2$ then $\left\|(\Pi_1\otimes \Pi_2)(\omega)\right\|_{X_1\otimes_\pi X_2} \leq (u_1\otimes u_2)(\omega)$;
\item[(b)] if $\omega\in C_1\tmax C_2$ then $\left\|(\Pi_1\otimes \Pi_2)(\omega)\right\|_{X_1\otimes_\epsilon X_2} \leq (u_1\otimes u_2)(\omega)$.
\end{enumerate}
Moreover, for all $z\in X_1\otimes X_2$,
\begin{enumerate}
    \item[(c)] if $\|z\|_{X_1\otimes_\pi X_2}\leq 1$ then $\gamma_1\otimes \gamma_2 + z\in C_1\tmin C_2$;
    \item[(d)] if $\|z\|_{X_1\otimes_\varepsilon X_2}\leq 1$ then $\gamma_1\otimes \gamma_2 + z\in C_1\tmax C_2$.
\end{enumerate}
\end{lemma}

\begin{proof}
We start with claim~(a). Let $\omega\in C_1\tmin C_2$ be decomposed as $\omega = \sum_j v_j\otimes w_j$, where $v_j\in C_1$ and $w_j\in C_2$. By~\eqref{symmetric cone reconstruction}, we have that $\left\|\Pi_1(v_j)\right\|_{X_1}\leq u_1(v_j)$ and $\left\|\Pi_2(w_j)\right\|_{X_2}\leq u_2(w_j)$ for all $j$. Using~\eqref{proj}, we deduce that
\begin{align*}
\left\|(\Pi_1\otimes \Pi_2)(\omega) \right\|_{X_1 \otimes_\pi X_2} &= \left\|\sumno_j \Pi_1(v_j)\otimes \Pi_2(w_j) \right\|_{X_1 \otimes_\pi X_2} \\
&\leq \sum_j  \left\|\Pi_1(v_j)\right\|_{X_1} \left\|\Pi_2(w_j)\right\|_{X_2} \\
&\leq \sum_j u_1(v_j)\, u_2(w_j) \\
&= \left( u_1\otimes u_2\right) (\omega)\, .
\end{align*}
We now move on to~(c). By closedness of $C_1\tmin C_2$, we can assume without loss of generality that $\|z\|_{X_1\otimes_\pi X_2} < 1$ holds with strict inequality. By~\eqref{proj}, there exists a decomposition $z=\sum_i x_i \otimes y_i$ such that $\sum_i \|x_i\|_{X_1} \|y_i\|_{X_2} < 1$. We now write the following explicitly separable decomposition for $\gamma_1\otimes \gamma_2 + z\in C_1\tmin C_2$:
\begin{equation*}
\begin{aligned}
    \gamma_1\otimes \gamma_2 + z &= \left(1-\sumno_i \|x_i\|_{X_1}\|y_i\|_{X_2}\right) \gamma_1\otimes \gamma_2 \\
    &\quad + \sum_i \frac{\|x_i\|_{X_1}\|y_i\|_{X_2}}{2} \bigg(\left(\gamma_1 + \frac{x_i}{\|x_i\|_{X_1}} \right)\! \otimes\! \left(\gamma_2 + \frac{y_i}{\|y_i\|_{X_2}} \right) \\
    &\qquad\qquad + \left(\gamma_1 - \frac{x_i}{\|x_i\|_{X_1}} \right)\!\otimes\! \left(\gamma_2 - \frac{y_i}{\|y_i\|_{X_2}} \right)\bigg)\, .
\end{aligned}
\end{equation*}
This proves claim~(c).

Claims~(b) and~(b) follow by duality. We start with~(b). Since $(V_i^*,C_i^*,\gamma_i)$ are symmetric GPTs, applying~(c) to them shows that for all $h\in X_1^*\tmin X_2^*$ with $\left\|h\right\|_{X_1^*\otimes_\pi X_2^*} \leq 1$ it holds that $u_1\otimes u_2 - h \in C_1^* \tmin C_2^*$. Hence, all $\omega\in C_1\tmax C_2 = \left(C_1^*\tmin C_2^*\right)^*$ are such that
\begin{equation*}
    0\leq (u_1\otimes u_2 - h)(\omega) = (u_1\otimes u_2)(\omega) - h(\omega) = (u_1\otimes u_2)(\omega) + h \left( (\Pi_1\otimes \Pi_2)(\omega)\right) .
\end{equation*}
Optimising over all $h\in B_{X_1^*\otimes_\pi X_2^*}$ and applying~\eqref{dual-norm} as well as the duality formula~\eqref{dual proj is inj} yields the inequality in~(b).

Thanks to~\eqref{dual min is max}, to prove~(d) it suffices to verify that $\xi(\gamma_1\otimes \gamma_2 + z)\geq 0$ for all $\xi\in C_1^*\tmin C_2^*$. Indeed,
\begin{align*}
    \xi(\gamma_1\otimes \gamma_2 + z) &= \xi(\gamma_1\otimes \gamma_2) + \xi\left((\Pi_1\otimes \Pi_2)(z)\right) \\
    &= \xi(\gamma_1\otimes \gamma_2) + \left( (\Pi_1^*\otimes \Pi_2^*)(\xi)\right) (z) \\
    &\textgeq{(i)} \xi(\gamma_1\otimes \gamma_2) - \left\|(\Pi_1^*\otimes \Pi_2^*)(\xi) \right\|_{(X_1\otimes_\varepsilon X_2)^*} \|z\|_{X_1\otimes_\varepsilon X_2} \\
    &\texteq{(ii)} \xi(\gamma_1\otimes \gamma_2) - \left\|(\Pi_1^*\otimes \Pi_2^*)(\xi) \right\|_{X_1^*\otimes_\pi X_2^*} \|z\|_{X_1\otimes_\varepsilon X_2} \\
    &\textgeq{(iii)} \xi(\gamma_1\otimes \gamma_2) - \left\|(\Pi_1^*\otimes \Pi_2^*)(\xi) \right\|_{X_1^*\otimes_\pi X_2^*} \\
    &\textgeq{(iv)} 0 \, ,
\end{align*}
where (i) is an application of~\eqref{dual-norm} to the normed space $X_1\otimes_\varepsilon X_2$, (ii) descends from~\eqref{dual inj is proj}, (iii) holds by hypothesis, and finally (iv) is just the inequality in~(a) stated at the dual level. 
\end{proof}

The above Lemma~\ref{lemma:norms-cones} gives us a natural way to construct candidate entangled tensors when the local theories are symmetric. For $z\in X_1\otimes X_2$ with $\|z\|_{X_1 \otimes_\varepsilon X_2}\leq 1$, the state
\begin{equation} \label{omega-z}
\omega(z) \coloneqq \gamma_1\otimes \gamma_2 + z
\end{equation}
satisfies $\omega(z)\in C_1\tmax C_2$.

\subsection{Proof of Theorem~\ref{symmetric result}}

We start by proving Lemma~\ref{lemma:ent-rob} as stated in the main text. This allows us to connect the entanglement robustness of states of the form~\eqref{omega-z} with the projective/injective tensor norm ratio of the parent tensor.

\begin{manuallemma}{\ref{lemma:ent-rob}}
Let $(V_1,C_1,u_1),\, (V_2,C_2,u_2)$ be two symmetric GPTs. Call $\gamma_1,\gamma_2$ the centres of the state spaces, and $X_1,X_2$ the associated normed spaces. For $z \in X_1\otimes X_2$, consider the normalised state $\omega(z)\coloneqq \gamma_1\otimes \gamma_2 + z$. Whenever $z$ satisfies $\|z\|_{X_1\otimes_\varepsilon X_2}\leq 1$, it holds that $\omega(z)\in C_1\tmax C_2$. In this case,
\begin{equation} \label{ent rob estimate}
\erob \left(\omega(z)\right) \geq \frac{\|z\|_{X_1\otimes_\pi X_2}-1}{2}\, .
\end{equation}
\end{manuallemma}

\begin{proof}
The first claim is just Lemma~\ref{lemma:norms-cones}(d). We now focus on the second. Let $\zeta \in V_1\otimes V_2$ be such that $\zeta,\, \omega(z) + \zeta \in C_1\tmin C_2$. Then,
\begin{align*}
    \|z\|_{X_1\otimes_\pi X_2} - (u_1\otimes u_2)(\zeta) &= \|(\Pi_1\otimes \Pi_2)(\omega(z))\|_{X_1\otimes_\pi X_2} - (u_1\otimes u_2)(\zeta) \\
    &\textleq{(i)} \|(\Pi_1\otimes \Pi_2)(\omega(z))\|_{X_1\otimes_\pi X_2} - \|(\Pi_1\otimes \Pi_2)(\zeta)\|_{X_1 \otimes_\pi X_2} \\
    &\textleq{(ii)} \|(\Pi_1\otimes \Pi_2)(\omega(z) + \zeta)\|_{X_1\otimes_\pi X_2} \\
    &\textleq{(iii)} (u_1\otimes u_1)\left( \omega(z) + \zeta \right) \\
    &\texteq{(iv)} 1+(u_1\otimes u_2)(\zeta)\, .
\end{align*}
The above steps are justified as follows: (i) follows from Lemma~\ref{lemma:norms-cones}(a) applied to $\zeta\in C_1\tmin C_2$; (ii) is just an application of the triangle inequality; (iii) is again Lemma~\ref{lemma:norms-cones}(a), this time applied to $\omega(z)+\zeta \in C_1\tmin C_2$; and (iv) descends from the easily verified fact that $\omega(z)$ as defined in~\eqref{omega-z} is normalised.

The above reasoning implies that $2(u_1\otimes u_2)(\zeta)\geq \|z\|_{X_1\otimes_\pi X_2} - 1$. Taking the infimum over all $\zeta$ and using the definition Eq.~\eqref{ent-rob} yields precisely Eq.~\eqref{ent rob estimate}. This concludes the proof.
\end{proof}

We are finally ready to prove Theorem~\ref{symmetric result}.

\begin{proof}[Proof of Theorem~\ref{symmetric result}]
Consider a pair of symmetric GPTs of dimensions $n+1,\,m+1\geq 3$. Call $X_1,\, X_2$ the normed spaces associated with them via~\eqref{X-space} and~\eqref{NX}. Since $\dim X_1=n$ and $\dim X_2=m$, applying Definition~\ref{def pi/epsilon ratio} we see that $\rho(X_1,X_2)\geq r(n,m)$. Now, remember from Fact~\ref{fact:inj<proj} that $\rho(X_1,X_2)$ is the \emph{smallest} constant such that $\|z\|_{X_1\otimes_\pi X_2}\leq \rho(X_1,X_2)\|z\|_{X_1\otimes_\varepsilon X_2}$ holds for all $z\in X_1\otimes X_2$. Hence, by compactness, there must exist a tensor $z_0$ that satisfies this inequality with equality.
Up to a multiplicative constant, we can assume without loss of generality that $\|z_0\|_{X_1\otimes_\varepsilon X_2}=1$ and hence that $\|z\|_{X_1\otimes_\pi X_2} = \rho(X_1,X_2)$. Then, the estimate in~\eqref{ent rob estimate} ensures that
\begin{equation}
    \erob(\omega(z_0)) \geq \frac{\|z_0\|_{X_1\otimes_\pi X_2} - 1}{2} \geq \frac{r(n,m)-1}{2}\, .
\end{equation}
The claims then follow from Facts~\ref{fact:19/18} and~\ref{fact:r(n,m)}, in turn derived from~\cite{XOR}.
\end{proof}

\label{section:symmetric}

\section{More about retracts} \label{section:retracts}



\begin{proof}[Proof of Proposition~\ref{proposition:2-dimensional-retract}]
We show that there exists a projection $P :V \to E$ such that $P(\C) = \C \cap E$. Let $x$ and $y$ be generators of the $2$ extremal rays of the $2$-dimensional cone $\C \cap E$.
Denote by $T_x$ and $T_y$ tangent hyperplanes to $\C$ at $x$ and $y$. We define $P$ by $\ker (P) = T_x \cap T_y$. Consider any element $z \in \C$ ; a (projective) geometric argument in the affine plane generated by $x$, $y$ and $z$ (see Figure~\ref{2-dim-retract}) shows that $P(z) \in \C \cap E$.
\end{proof}

\begin{figure}[htbp] \begin{center}
\begin{tikzpicture}[scale=1]
\draw (0,0) ..controls +(-0.1,1) and +(-3,1).. (2.5,2); 
\draw (2.5,2) ..controls +(3,-1) and +(2,1).. (3.5,-1); 
\draw (3.5,-1) ..controls +(-2,-1) and +(0.1,-1).. (0,0); 
\draw (2.5,2) node {$\bullet$};
\draw (2.5,2) node [above right] {$x$};
\draw (-0.5,3) -- (7.5,1/3) ;
\draw (3.5,-1) node {$\bullet$};
\draw (3.5,-1) node [below right] {$y$};
\draw (1.5,-2) -- (7.5,1) ;
\draw (3.7,0) node {$\bullet$};
\draw (3.7,0) node [above] {$z$};
\draw[thick] (2,3.5) -- (4,-2.5) ;
\draw[dashed] (7.7,0.8) -- (-0.3,-0.8) ;
\draw (-0.5,3) node [above] {$T_x$} ;
\draw (1.5,-2) node [left] {$T_y$} ;
\draw (1,1) node {$\C$} ;
\draw (2,3.5) node [left] {$E$} ;
\end{tikzpicture}
\end{center}
\caption{The image of $z$ under $P$ is a positive multiple of the intersection between $E$ and the (dashed) line through $z$ and $T_x \cap T_y$, and therefore belongs to $\C$.}
\label{2-dim-retract}
\end{figure}
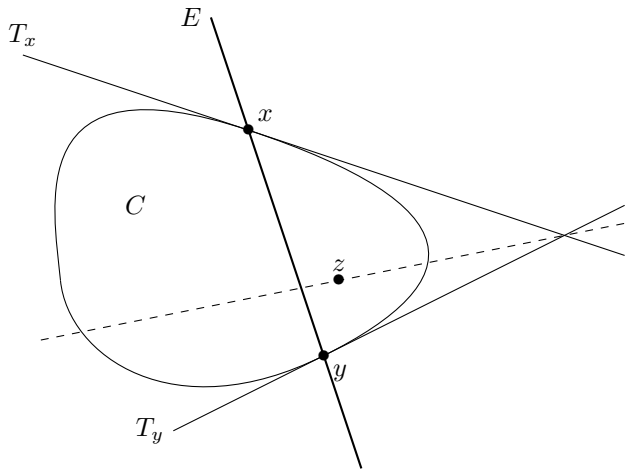

One could be tempted to conjecture that every cone has some nontrivial retract. However, this is remarkably not the case, as we show now. Our argument relies heavily on a notable result by Zamfirescu~\cite{Zamfirescu1991}. Before we delve into the details, we need to introduce some notation. Let $\mathcal{K}_n$ denote the family of all convex bodies (compact convex sets with nonempty interior) in $\R^n$. We can endow $\mathcal{K}_n$ with the \emph{Hausdorff distance}, defined as
\begin{equation}
    d_H(X,Y)\coloneqq \inf\left\{\epsilon>0:\ X\subseteq Y_\epsilon,\ Y\subseteq X_\epsilon \right\}
\end{equation}
for all $X,Y\in\mathcal{K}_n$, where for a convex body $K\in \mathcal{K}_n$ we denoted by $K_\epsilon\coloneqq \left\{ x\in \R^n:\ d(x,K)\leq \epsilon \right\}$ its `$\epsilon$-fattening'. Here, $d(x,K)$ quantifies the distance of $x$ from $K$ as measured by the Euclidean norm. It is well known that $\mathcal{K}_n$ equipped with $d_H$ becomes a complete metric space~\cite[Thm.~1.8.2,~1.8.5]{SCHNEIDER} and hence a Baire space. We remind the reader a property is said to be obeyed by most elements in a Baire space if those that violate it form a \emph{meagre set}, i.e.\ a countable union of sets whose closure has empty interior.

For $K\in\mathcal{K}_n$ and $x\notin K$, the \emph{shadow boundary} of $K$ with respect to $x$ is defined as
\begin{equation}
    \partial(K,x) \coloneqq \left\{ y\in K:\ \aff \{x,y\}\cap\inter(K)=\emptyset \right\} ,
\end{equation}
where $\aff$ denotes the affine hull, and hence $\aff\{x,y\}$ is nothing but the straight line through $x$ and $y$. It is easily verified that the above definition can be straightforwardly extended to the more general case where $x\in \P^n\setminus K$, where $\P^n$ denotes the real projective space of dimension $n$.\footnote{A useful way to think about $\P^n$ is as follows. Consider a hyperplane $H_n$ of dimension $n$ in an $(n+1)$-dimensional space. Pick a point $p\notin H_n$. While $\R^n$ can be identified with $H_n$, the projective space $\P^n$ can be thought of as the set of straight lines through $p$. The natural embedding $\R^n\subset \P^n$ can be obtained by noticing that those lines through $p$ that are not parallel to $H_n$ identify a unique point on it.}
The result by Zamfirescu~\cite[Thm.~1]{Zamfirescu1991} asserts that for most convex bodies in $\mathcal{K}_n$ it holds that
\begin{equation}
    \dim \aff \partial(K,x) = n\qquad \forall\ x\in\P^n\setminus K\, .
    \label{Zamfirescu eq}
\end{equation}

The following lemma shows that having retracts is an exceptional property. In particular, this shows that the conclusion of Proposition~\ref{proposition:baire-no-retract} is true for most convex cones.

\begin{lemma} \label{Zamfirescu lemma}
For all $n\geq 4$, most convex bodies $K\in\mathcal{K}_{n-1}$ are such that the cone $\CC(K)$ constructed via~\eqref{eq:cone-over-K} has no $(n-1)$-dimensional retracts.
\end{lemma}

\begin{proof}
Most convex bodies obey~\eqref{Zamfirescu eq} by Zamfirescu's theorem. Moreover, it is known that most convex bodies are strictly convex~\cite[Thm.~2.6.1]{SCHNEIDER}, meaning that their boundary contains no nontrivial segment. Hence, most convex body are \emph{simultaneously} strictly convex \emph{and} obey~\eqref{Zamfirescu eq}. We proceed to show that any such convex body $K\in\mathcal{K}_{n-1}$ defines via~\eqref{eq:cone-over-K} an $n$-dimensional cone $\CC(K)$ that has no $(n-1)$-dimensional retracts. By Fact~\ref{fact:retract-as-section}, is it enough to show that there is no projection $P:\R^n \to V$ of rank $n-1$ satisfying $P(\CC(K))=\CC(K) \cap V$.

Consider $K$ as embedded in the affine subspace $H$ defined by setting the last coordinate of $\R^n$ to $1$. Assume that $\CC(K)$ admits a retract to an $(n-1)$-dimensional subspace $V$, and set $W\coloneqq V\cap H$. Call $P:\R^{n}\to V$ a projection that satisfies $P(\CC(K)) = \CC(K)\cap V$. Clearly, $\dim \ker P=1$, i.e.\ $\ker P$ is a straight line. Consider the point $x\in \P^{n-1}$ such that $\{x\} = H\cap \ker P$ (observe that $x$ can be at infinity, when $\ker P$ is parallel to $H$). Then we claim that: (i) $x\notin K$; and (ii) $\partial(K,x)\subseteq W$. This contradicts the assumption that $\dim \aff \partial (K,x)=n-1$.

To prove (i), assume that $x\in K = \inter (K) \cup \partial K$. If $x\in \inter(K)$, then it must be that $P\equiv 0$, which is naturally absurd. In fact, if $Py\neq 0$ for some $y\in\R^n$, using the fact that $ty+x\in \CC(K)$ for all $t\in (-\epsilon,\epsilon)$ (where $\epsilon>0$), we would obtain that $P(ty+x) = t Py \in P(\CC(K))=\CC(K)\cap V$ for $t$ in a neighbourhood of $0$, which is in contradiction with $\CC(K)$ being proper. If $x\in \partial K$, take $y\in \partial (K\cap W) \subseteq \partial K \cap W$. Since for all $\lambda\in [0,1]$ one has that $P\left(\lambda x+(1-\lambda)y\right) = (1-\lambda) y\in \partial \CC(K)$, and since $P(\CC(K))\subseteq \CC(K)$, we deduce that $\lambda x+(1-\lambda)y\in \partial \CC(K)\cap H = \partial K$ for all $\lambda\in [0,1]$; this contradicts the assumption that $K$ is strictly convex.

We now move on to (ii). Since $x\notin K$, the shadow boundary $\partial(K,x)$ is nonempty. Assume by contradiction that there is $y\in \partial(K,x)\setminus W$. This means that $y-Py = \mu x$ for some real $\mu\neq 0$. For $\lambda\in\R$, consider the point
\begin{equation*}
    z_\lambda \coloneqq \lambda x + (1-\lambda) y = \left( 1 - \lambda + \frac{\lambda}{\mu} \right) y - \frac{\lambda}{\mu} Py\, .
\end{equation*}
It is not difficult to check that since $\mu\neq 0$ one can satisfy both $1 - \lambda + \lambda/\mu \geq 0$ and $\lambda/\mu\leq 0$ for all $\lambda$ in a nontrivial left- or right-neighbourhood of $0$. For such values of $\lambda$ one obtains that $z_\lambda\in \CC(K)$, and since $z_\lambda\in H$ by construction it holds in fact that $z_\lambda\in K$. We have shown that there is a nontrivial segment in the straight line $\aff\{x,y\}$ (one of whose extremes is $y$) that is entirely contained in $K$. Since $K$ is strictly convex, it cannot be that this segment is entirely contained in the boundary $\partial K$. Hence, it must intersect the interior $\inter (K)$. This contradicts the assumption that $y\in\partial (K,x)$, and concludes the proof.
\end{proof}


\end{document}